\documentclass[10pt,twocolumn,twoside]{IEEEtran}

\ifCLASSINFOpdf
\else
\fi

\usepackage{graphicx}

\usepackage{url}
\usepackage{subfigure}        
\usepackage[usenames]{color}            
\usepackage{amsfonts}%
\usepackage{amssymb}%
\usepackage{mathrsfs}
\usepackage{amsmath}%
\usepackage{cite}%

\usepackage[latin1]{inputenc}

\newtheorem{theorem}{Theorem}
\newtheorem{proposition}{Proposition}
\newtheorem{lemma}{Lemma}

\newcommand{\EE}{\mathbf{E}}

\begin{document}
\title{Zero-Delay Joint Source-Channel Coding for a Multivariate Gaussian on a Gaussian MAC}

\author{Pål~Anders~Floor, 
        Anna~N.~Kim,
        Tor~A.~Ramstad,
        Ilangko~Balasingham,
        Niklas~Wernersson, Mikael~Skoglund,
\thanks{P. A. Floor  and I. Balasingham are with the Interventional Center,
Oslo University Hospital and Institute of Clinical Medicine, University of Oslo,
  Oslo, Norway (e-mail: andflo@rr-research.no). A. N. Kim is with SINTEF ICT, Oslo, Norway (e-mail: annak@ieee.org).  P. A. Floor, T. A.
Ramstad  and I. Balasingham are with the Department of Electronics and
Telecommunication, NTNU, Trondheim,
Norway.}
\thanks{Niklas Wernersson is with Ericsson Research, Stockholm, Sweden.
Mikael Skoglund is with the School of Electrical Engineering and the ACCESS Linnaeus Center,
at the Royal Institute of Technology (KTH), Stockholm, Sweden}
}
\maketitle

\begin{abstract}
In this paper, communication of a Multivariate Gaussian over a Gaussian Multiple
Access Channel is studied. Distributed zero-delay joint source-channel coding (JSCC) solutions to the problem are given. Both nonlinear and linear approaches are discussed. The performance upper bound (signal-to-distortion ratio) for arbitrary code length is also derived and Zero-delay cooperative JSCC is briefly addressed in order to provide an
approximate bound on the performance of zero-delay schemes.
  The main contribution is a nonlinear
hybrid discrete-analog JSSC scheme based on distributed quantization and
a linear continuous mapping named Distributed Quantizer Linear
Coder (DQLC). The DQLC has promising performance which improves
with increasing correlation, and is robust against variations in
noise level.  The DQLC exhibits a constant gap to the performance upper bound as the signal-to-noise ratio (SNR) becomes large for any number of sources and values of correlation. Therefore it outperforms a linear solution (uncoded transmission) in any case when the SNR gets sufficiently large.
\end{abstract}
\IEEEpeerreviewmaketitle
\section{Introduction}\label{sec:intro}
In this paper we investigate joint source-channel coding (JSCC) for a multipoint-to-point problem, where
multiple memoryless and inter-correlated Gaussian sources are
transmitted distributedly over a memoryless Gaussian multiple access channel (GMAC) without multiplexing.
There are mainly two cases to consider for this network:
1) Recovery of the common information shared by all sources. 2) Recovery
of each individual source. In Case 1), when the transmit power of all sources are equal, the distortion lower bound can be achieved by a simple zero-delay linear mapping, often referred to as \emph{uncoded transmission}~\cite{Gastpar08}.


In case 2) both common information as well as the individual variations of each source are reconstructed at the receiver.
The bivariate case (two sources) was treated in~\cite{Lapidoth10} for arbitrary codeword length. The authors proposed a nonlinear hybrid JSCC scheme that superimposes a rate optimal (infinite dimensional) vector quantizer with uncoded transmission. This hybrid scheme was shown to achieve the distortion lower bound at high and low channel signal-to-noise ratio (SNR), while a small gap remains for other SNR. Contrary to case 1), optimality of uncoded transmission is restricted (to low SNR in general), and infinite complexity and delay JSCC are required to get close to the distortion lower bound in general. The bivariate case was also treated in~\cite{Floor_Kim_Wernersson11_TCOM}, where a zero-delay (single letter) constraint was imposed in order to provide simple low complexity and possibly implementable solutions to the problem. Two zero-delay nonlinear schemes were proposed, one discrete scheme and one hybrid discrete-analog scheme. Both schemes perform well and can improve significantly on uncoded transmission when the channel signal-to-noise ratio (SNR) is sufficiently large. Nevertheless, there remains a constant gap to the distortion lower bound for all SNR due to the zero-delay constraint.

The main contribution of this paper is to provide a nonlinear zero-delay JSCC solution for the multivariate version of case 2), since there to our knowledge exists few, if any, results on this in the literature. We generalize the hybrid discrete-analog scheme from~\cite{Floor_Kim_Wernersson11_TCOM} (named \emph{SQLC})
to the multivariate case. The proposed scheme is named \emph{DQLC} since it consists of distributed quantizers and a linear continuous mapping. The main advantage of DQLC is that
it provides a simple JSCC solution to the problem that can significantly outperform uncoded transmission
under sufficiently large channel SNR.

To assess the performance of DQLC, the distortion lower bound for the problem under consideration is derived. Since this lower bound assumes arbitrary codelength, one must expect a significant backoff from it when assessing performance of a zero-delay JSCC scheme. To provide indications on the bound for zero-delay distributed schemes, and thereby a better measure on how well the DQLC perform, known zero-delay schemes with cooperative encoders\footnote{By cooperation we mean that all source symbols are available at all encoders without any additional use of resources.} are addressed.  It is also shown that DQLC exhibits a constant gap to the distortion lower bound (or performance upper bound, the bound on signal-to-distortion ratio) as SNR $\rightarrow\infty$. This gap increases somewhat with the number of of sources, but remains bounded as the number of sources becomes large. Since uncoded transmission has a \emph{leveling off} effect at high SNR, DQLC will therefore outperform uncoded transmission for any number of sources when the SNR is sufficiently large.

The paper is organized as follows: In
Section~\ref{sec:probst_ub}, a problem formulation is given and the distortion lower bound is derived. Zero-delay cooperative encoding and uncoded transmission are also introduced. In Section~\ref{sec:DQLC} the DQLC is introduced and its distortion is derived mathematically. It is further shown that DQLC exhibits a constant gap to distortion lower bound as SNR$\rightarrow\infty$.
In Section~\ref{sec:Ex_SNQLC_M3} we concentrate on the 3 source case and optimize the DQLC for general SNR. Its performance is compared to the derived bounds and uncoded transmission. A summary and brief discussion are given in Section~\ref{sec:summary}.

Note that some of the results in this paper have previously been published in~\cite{Floor_kim_ramstad_ISABEL2011}. 

\section{Problem statement and bounds}\label{sec:probst_ub}
Fig.~\ref{fig:sgsn}
depicts the communication system under consideration.
\begin{figure}[h]
\centering
  \includegraphics[width=1\columnwidth]{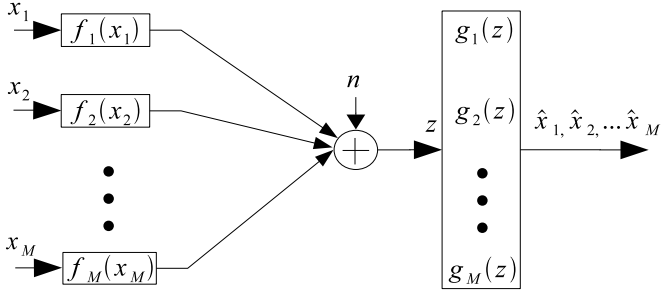}
  \caption{Network under consideration. A multivariate Gaussian is communicated on a Gaussian MAC with separate (distributed) encoders. 
  }\label{fig:sgsn}
\end{figure}

\subsection{Problem statement}\label{ssec:probst}
The sources are memoryless discrete time, continuous
amplitude, zero mean Gaussian random variables
$x_1,...,x_M$, where the model $x_m=s+w_m,
m\in\{1,...,M\}$ is considered. That is, the sources share a common information
$s\sim\mathcal{N}(0,\sigma_s^2)$, while each source also has an individual
component $w_m\sim\mathcal{N}(0,\sigma_{w_m}^2)$. A simplified scenario $\sigma_{w_1}=\sigma_{w_2}=\cdots = \sigma_{w_M}$ is assumed here, which
implies that $\sigma_{x_1}=\sigma_{x_2}=\cdots =
\sigma_{x_M}=\sigma_x$, making it easier to derive theoretical expressions that can be analyzed further. The correlation between any two sources is
then $\rho_{ij}=\EE\{x_i x_j\}/\sigma_x^2=\rho_x, \forall i,j$, resulting in a simple
covariance matrix $K_\mathbf{x}=\EE\{\mathbf{x}\mathbf{x}^T\}$ with
$\sigma_x^2$ on the diagonal and $\sigma_x^2\rho_x$ in every
off-diagonal element, and eigenvalues $\lambda_1= \sigma_x^2 ((M-1)\rho_x+1)$ and $\lambda_i= \sigma_x^2 (1-\rho_x)$, $i=2,\cdots,M$. Note that the schemes presented in this paper can be applied for any case of unequal correlations $\rho_{ij}$. The mathematical analysis becomes more complicated, however.

Each encoder consists of a memoryless
encoding function $y_m = f_m(x_i), m\in\{1,...,M\},$ and its output is transmitted
over GMAC with additive noise $n\sim\mathcal{N}(0,\sigma_{n}^2)$.
The received signal is
\begin{equation}\label{e:rec_sign}
z=\sum_{m=1}^M f_m(x_m)+n \in \mathbb{R}.
\end{equation}
At the receiver the functions $g_m(z), m\in\{1,...,M\}$,
produce the estimates $\hat{x}_1,...,\hat{x}_M$ from the channel
output $z$. We define the end-to-end distortion $D$ as the
mean-squared-error (MSE) averaged over \textit{all} source symbols
\begin{equation}\label{e:tot_dist}
D = \frac{1}{M}\left(\EE\{|x_1-\hat{x}_1|^2\}+\cdots+\EE\{|x_M-\hat{x}_M|^2\}\right).
\end{equation}
The average transmit power for node $m$ is defined as $P_m=\EE\{f_m(x_m)^2\}$. 
We further assume ideal Nyquist sampling and an ideal Nyquist
channel, where the sampling rate of each source is the same as the
signalling rate of the channel. We also assume ideal synchronization
and timing between all nodes. Our design objective is to find the
$f_m$ and $g_m$ that minimizes $D$. The DQLC encoders are asymmetric, where $P_1\geq P_2\geq\cdots \geq P_M$. Therefore, an average transmit
power constraint $P=(P_1+P_2+\cdots P_M)/M$ is considered.

\subsection{Distortion lower bound}\label{ssec:Dist_bound}
When no collaboration is allowed among the encoders, the achievable distortion bound is unknown. One can, however, derive a lower bound
by considering the ideal scenario with full collaboration among
all encoders at no additional cost. This scenario is then a point-to-point communication problem, where the distortion lower bound can be determined by equating the rate-distortion function for an
$M$ dimensional Gaussian source to the rate of the GMAC. The following proposition quantifies this bound.
\begin{proposition}
The distortion lower bound for the network in Fig.~\ref{fig:sgsn}, is in the \emph{symmetric case} $D_1=D_2=\cdots =D_M=D$ with transmit power $P_1=P_2=\cdots=P_M=P$ and correlation
$\rho_{ij}=\rho_x$, $\forall i,j$, given by
\begin{align}\label{e:Dist_bound_equalP}
{D}\geq
\begin{cases}
\sigma_x^2\left(1-\frac{P(1+(M-1)\rho_x)^2}{P(M^2\rho_x+M(1-\rho_x))+\sigma_n^2}\right)& \frac{P}{\sigma_n^2} \in\big(0,\frac{\rho_x}{1-\rho_x^2}\big],\vspace{0.3cm}\\
\sigma_x^2 \sqrt[M]{\frac{\big(1+(M-1)\rho_x\big)(1-\rho_x)^{M-1}\sigma_n^2}{P\big(M+M(M-1)\rho_x\big)+\sigma_n^2}}, & \frac{P}{\sigma_n^2}> \frac{\rho_x}{1-\rho_x^2}.
\end{cases}
\end{align}
\end{proposition}
\begin{proof}
Let $R^*$, $D^*$ and $P^*$  denote optimal rate, distortion and power respectively.
Assuming full collaboration, the $M$ sources can be
considered as a Gaussian vector source of dimension $M$. From~\cite{cover06}
\begin{equation}\label{eqn_d_lowerbound}
D^*(\theta,M) = \frac{1}{M}\sum_{i=1}^M \min[\theta,\lambda_i],
\end{equation}
\begin{equation}\label{e:_vecGaussRD}
R^*(\theta,M) = \frac{1}{M}\sum_{i=1}^M
\max\left[0,\frac{1}{2}\log_2\frac{\lambda_i}{\theta}\right],
\end{equation}
where $\lambda_i$ is the $i$-th eigenvalue of the covariance matrix $K_{\mathbf{x}}$.

Assuming that all $M$ encoders transmit at the same power and that the correlation between the encoder outputs, $y_m$, are equal to
$\rho_x\geq0$, the following Lemma results
\begin{lemma}
The channel capacity per source symbol is
\begin{align}\label{eqn_chcapacity}
C =\frac{1}{2M}\log_2\bigg(1+\frac{MP^*(1+(M-1)\rho_x)}{\sigma_n^2}\bigg).
\end{align}
\end{lemma}
\begin{proof}
Let $Y_{i,k}$ be the i'th encoder output at time instant $k$. Using Lemma C.2 and Theorem C.1 from \cite{Lapidoth10}, it can be shown that:
\begin{align}
\frac{1}{n}&\sum_{k=1}^n \EE\left[\left(\sum_{i=1}^M Y_{i,k}\right)^2\right] = \sum_{i=1}^M\left(\sum_{k=1}^n\EE[Y_{i,k}^2]\right)+\nonumber\\ &\ \ \ \ \ \ \ \ \ \ \ \ 2\sum_{i,j=1,i\neq
j}^M\left(\frac{1}{n}\sum_{k=1}^n\EE[Y_{i,k}Y_{j,k}]\right)\nonumber\\ &\stackrel{(a)}{\leq}MP+2\rho_x\sum_{i,j=1,i\neq
j}^M\left(\frac{1}{n}\sum_{k=1}^n\sqrt{\EE[Y_{i,k}^2]}\sqrt{\EE[Y_{j,k}^2]}\right)\nonumber\\ &\stackrel{(b)}{\leq}MP+2\rho_x\sum_{i,j=1,i\neq
j}^M\left(\frac{1}{n}\sqrt{\sum_{k=1}^n\EE[Y_{i,k}^2]}\sqrt{\sum_{k=1}^n\EE[Y_{j,k}^2]}\right)\nonumber\\ &\leq MP+2\rho_x\dbinom{M}{2}P
\end{align}
where $(a)$ comes from~\cite[Lemma C.2]{Lapidoth10} given $\rho_x\geq 0$, and $(b)$ is the Cauchy-Schwartz inequality. The resulting channel capacity is then~(\ref{eqn_chcapacity}).
\end{proof}

Now equate $R^*$ from (\ref{e:_vecGaussRD}) with $C$ in
(\ref{eqn_chcapacity}) and calculate the corresponding power $P^*$.
We get $D\geq D^*(\theta, M)$ with $D^*$ given in~(\ref{eqn_d_lowerbound}) and
\begin{equation}\label{eqn_MP}
P=P^*(\theta,M) =\sigma_n^2\frac{\prod_1^M{\max[\lambda_i/\theta,1]}-1}{M+M(M-1)\rho_x}
\end{equation}
The max and min in~(\ref{eqn_d_lowerbound}) and~(\ref{eqn_MP}) depends on $\rho_x$ and the SNR.  Since the special case $\rho_{ij}=\rho_x, \forall i,j$ is treated, there are two cases to consider: Only the first eigenvalue $\lambda_1$ is to be represented (the common information) and all eigenvalues, $\lambda_i, i\in[1,\cdots M]$, are to be represented. The validity of these two cases is the same as in~\cite{Lapidoth10}, i.e. SNR=$P/\sigma_n^2 \leq \rho_x/(1-\rho_x^2)$ and SNR=$P/\sigma_n^2 > \rho_x/(1-\rho_x^2)$ respectively. If one solve~(\ref{eqn_MP}) with respect to $\theta$ for these two cases and insert the result in~(\ref{eqn_d_lowerbound}), the bound~(\ref{e:Dist_bound_equalP}) results.
\end{proof}

Note that~(\ref{e:Dist_bound_equalP}) becomes a bound for an average transmit power constraint by setting $P=(P_1+P_2+\cdots+P_M)/M$.


\subsection{Zero-delay JSCC with collaborating encoders}\label{ssec:colaborative_sk}
Since collaboration makes it possible to construct a larger set of encoding operations, including all distributed strategies, the performance of distributed
coding schemes is upper-bounded by those that allow collaboration. Optimal zero-delay collaborative schemes therefore serve as a tighter bound for single letter schemes compared to the ones without restrictions on codeword length.

In the case of zero delay, the corresponding optimal collaborative
encoding operation is the mapping $\mathbb{R}^M\rightarrow \mathbb{R}$
from source to channel space, which minimizes $D$ at a given power
constraint. It has not yet been determined how such a mapping should be constructed in order to perform optimally. One can anyway get an
indication on how a scheme with collaborative encoders performs from schemes that are known to operate close to the distortion lower bound. Examples on known schemes with excellent performance are
\emph{Shannon-Kotel'nikov mappings} (S-K mappings)~\cite{ramstad-telektronikk,hekland_floor_ramstad_T_comm,Floor_Ramstad09,Akyol_rose_ramstad_itw10,Hu_garcia_lamarca_tcom} and \emph{Power Constrained Channel
Optimized Vector
Quantizers} (PCCOVQ)~\cite{fulds97a,fuldsethThesis}.  PCCOVQ can be considered similar to S-K mappings when the number of centroids is large and is therefore referred to as S-K mappings in the following. S-K mappings have previously been optimized for
memoryless Gaussian sources and channels when  $M$ source symbols are
transmitted on $N$ channel uses~\cite{ramstad-telektronikk,Floor_Ramstad09,hekland_floor_ramstad_T_comm,Akyol_rose_ramstad_itw10,Hu_garcia_lamarca_tcom,fulds97a,fuldsethThesis}.

For the problem at hand, by treating the collaborative encoders as one, and the $M$ sources as components of
a Gaussian vector source, S-K mappings with $N=1$ can be applied directly. S-K mappings can not be applied to the
distributed case, however, as their operation would require knowledge of all
source symbols simultaneously at each encoder~\cite{hekland_floor_ramstad_T_comm,ramstad-telektronikk,Floor_Ramstad09,fulds97a,fuldsethThesis}.

\subsection{Distributed linear JSCC: uncoded transmission}\label{sec:linear}
A simple way to construct zero-delay distributed encoders is to let
each sensor node scale its observations to satisfy the power
constraint, i.e. $f_m(x_m) = x_m\sqrt{{P}/{\sigma_x^2}}$.
With equal transmit power, the received signal becomes
\begin{equation}
z =\sqrt{\frac{P}{\sigma_x^2}}\left(Ms + \sum_{m=1}^M w_m\right)+n.
\end{equation}
At the receiver, MMSE decoding given by $\hat{x}_m = (\EE[x_m z]/\EE[z^2])z$, is
applied. One can show that the resulting end-to-end distortion becomes
\begin{equation}
\begin{split}
D &=\sigma_x^2-\frac{\EE[x_iz]^2}{\EE[z^2]}\\ &= \sigma_x^2\left(1-\frac{P(1+(M-1)\rho_x)^2}{P(M^2\rho_x+M(1-\rho_x))+\sigma_n^2}\right).
\end{split}
\end{equation}
Considering~(\ref{e:Dist_bound_equalP}), one can see that a linear mapping is optimal when $P/\sigma_n^2\leq \rho_x/(1-\rho_x^2)$, i.e. when only the common
information $s$ can be reconstructed.
\section{Distributed nonlinear JSCC: DQLC}\label{sec:DQLC}
In order to get closer to the bound in~(\ref{e:Dist_bound_equalP}) than uncoded transmission when $P/\sigma_n^2 > \rho_x/(1-\rho_x^2)$, nonlinear mappings are needed. A zero-delay hybrid
discrete-analog scheme, DQLC, where encoders 1 to $M-1$ are amplitude limited scalar quantizers and encoder $M$ consist of a limiter
followed by scaling, is presented here.

\subsection{Formulation of Encoders and Decoders}\label{ssec:SNQLC_enc_dec_gen}
Encoder $m$, $m\in [1,\ldots,M-1]$, first quantizes
$x_m$, then limits the quantizer output to a certain
range $\pm\kappa_m$, where $\kappa_m\in\mathbb{R}^+$, and further attenuates the result by $\alpha_m$. That is $f_m(x_m)=\alpha_m\ell_{\pm \kappa_m}[q_{\Delta_m}(x_m)]$.
 $q_{\Delta_m}(x_m)$ is a uniform midrise or midthread quantizer, where $q_{i_m}$ denotes centroid no. $i$ of quantizer $m$. Encoder $M$ only limits $x_M$
to $\pm\kappa_M$, then attenuates the result by $\alpha_M$. That is
$f_M(x_M)=\alpha_M \ell_{\pm \kappa_M}[x_M]$.  The
received signal becomes
\begin{equation}\label{e:rec_sign_SNQLC}
\begin{split}
&z=\sum_{m=1}^{M-1}\alpha_m (\ell_{\pm \kappa_m}[q_{\Delta_m}(x_m)])+\alpha_M (\ell_{\pm \kappa_M}[x_M])+n.
\end{split}
\end{equation}
In the following we choose $\alpha_1=1$ and $\kappa_1 = \infty$.

An example for $M=3$ is shown in Fig.~\ref{fig:SNQLC_corr_concept}.
\begin{figure}[h!]
    \begin{center}
    \subfigure[]{
            \includegraphics[width=.46\columnwidth]{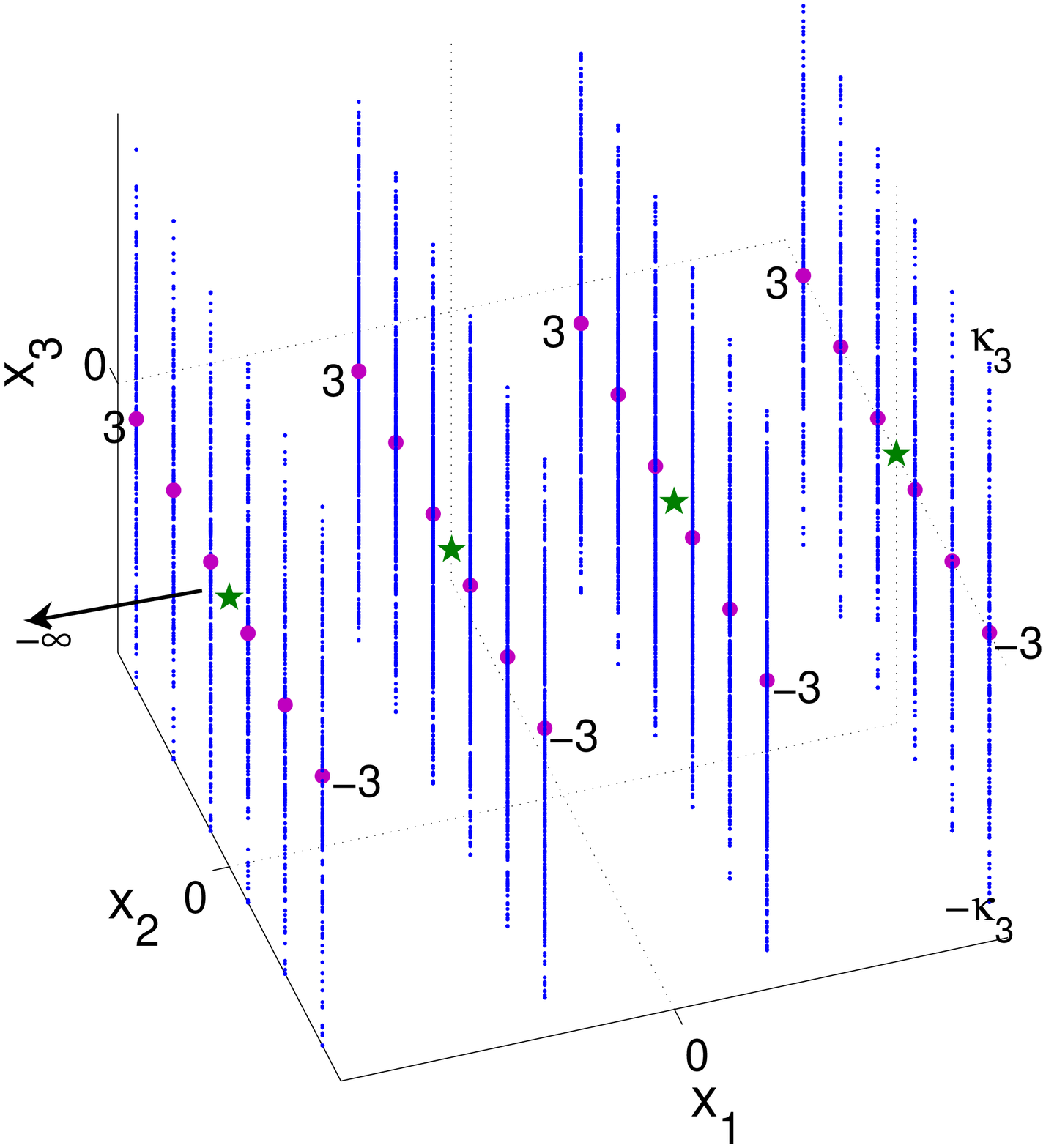}
        \label{fig:sqlc_scheme_source_r0}}
        \subfigure[]{
            \includegraphics[width=.46\columnwidth]{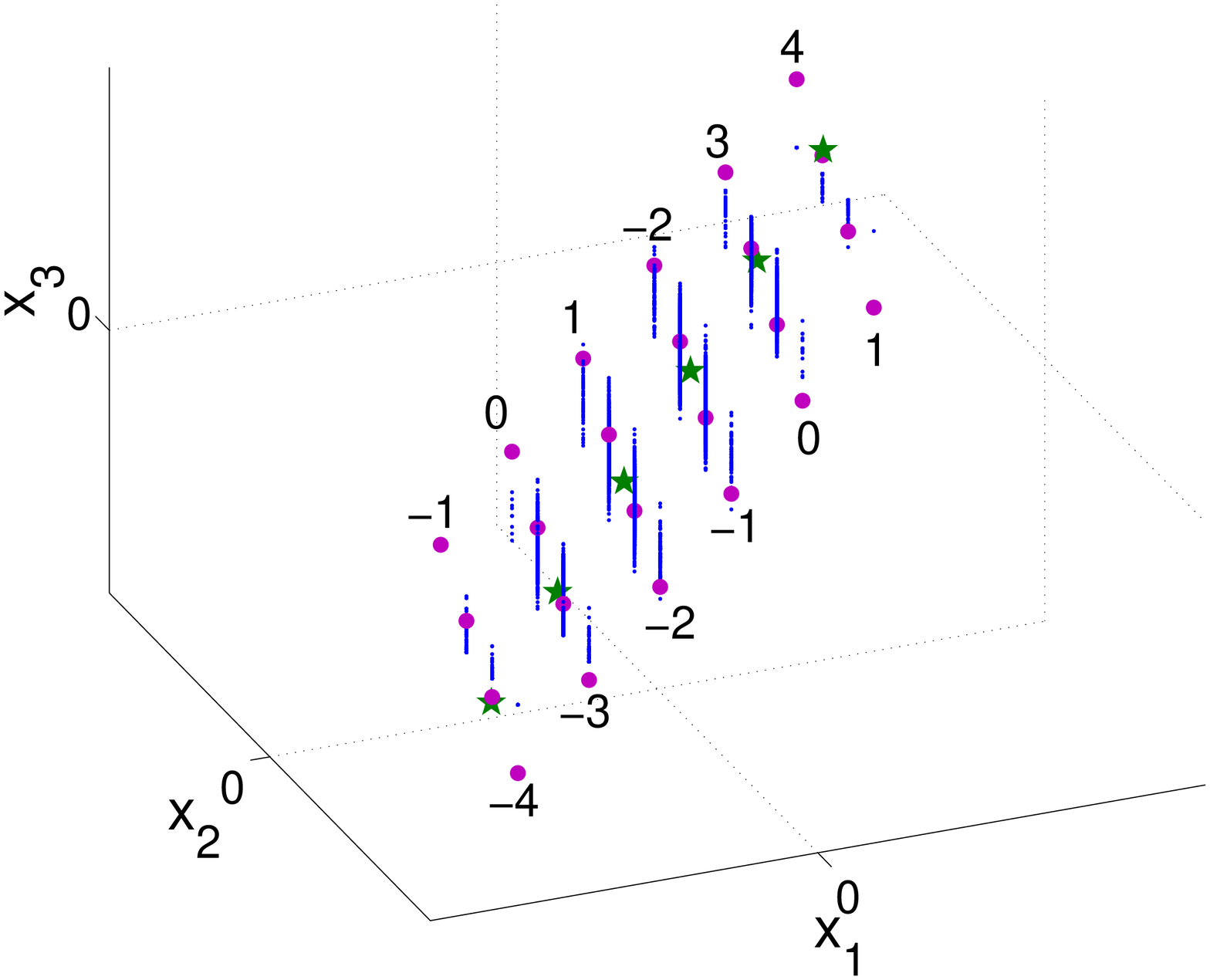}
        \label{fig:sqlc_scheme_source_r095}}
        \vspace{0.5cm}
        \subfigure[]{
            \includegraphics[width=0.9\columnwidth]{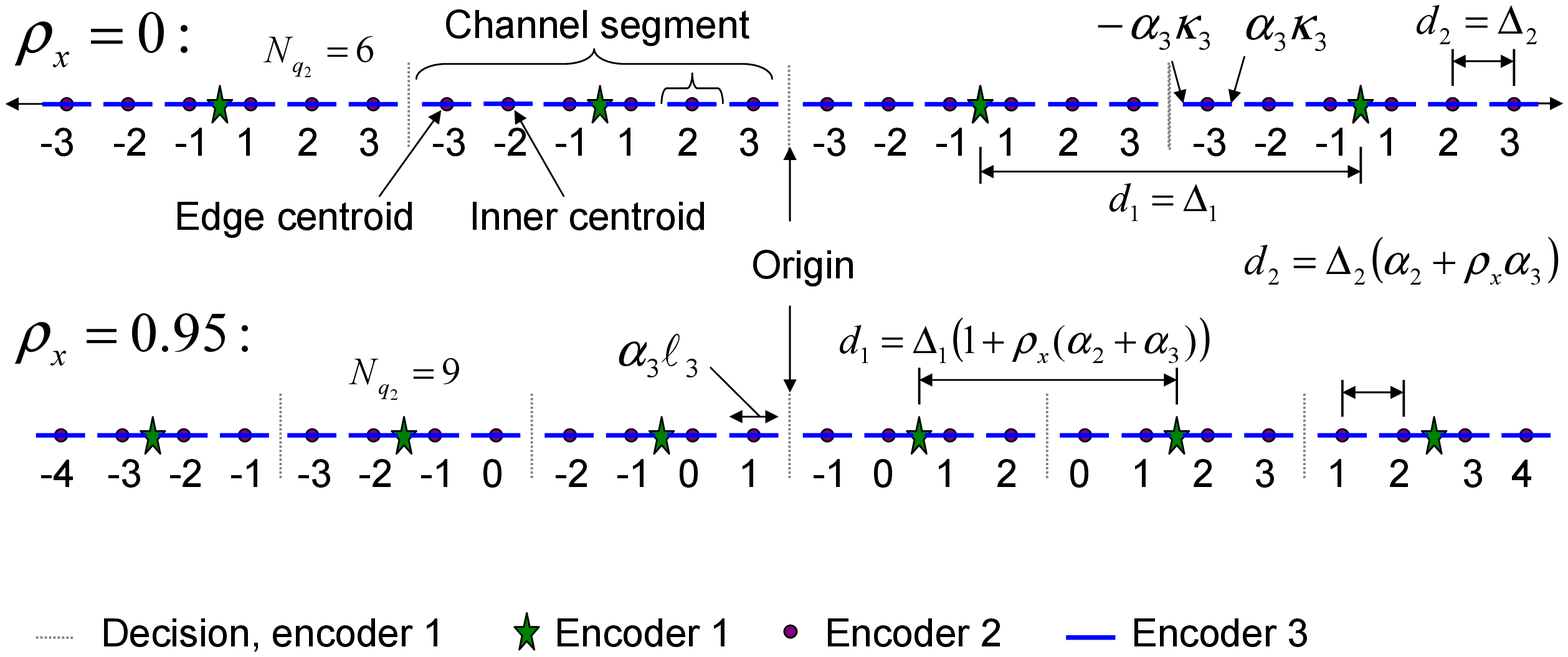}
        \label{fig:sqlc_scheme_channel}}
    \end{center}
    \caption{DQLC for $M=3$. Source space for: (a) $\rho_x=0$. (b) $\rho_x=0.95$. (c) Channel space. Here $\kappa_2=4$ when $\rho_x=0$ and $\kappa_2=5$ when $\rho_x=0.95$.}
     \label{fig:SNQLC_corr_concept}
\end{figure}
The encoders first create the segments in source space as shown in
Fig.~\ref{fig:sqlc_scheme_source_r0} ($\rho_x=0$)
and~\ref{fig:sqlc_scheme_source_r095} ($\rho_x=0.95$) through
quantization and limitation. These segments are attenuated by
$\alpha_m$ in such a way that the channel space structure shown in
Fig.~\ref{fig:sqlc_scheme_channel} results when GMAC sums over all
encoder outputs. To obtain this structure, $\alpha_1>\alpha_2>\cdots>\alpha_M$. From Fig.~\ref{fig:sqlc_scheme_source_r095} one can
see how DQLC is affected by correlation. As $\rho_x$ increases, the joint pdf
$p_\mathbf{x}(x_1,\cdots,x_M)$ narrows along all its minor axes,
effectively limiting each source segment. The operation $\ell_{\pm \kappa_m}$ then becomes obsolete. This effect results in reduced distortion. As
$\rho_x\rightarrow 1$ one can let $\Delta_m \rightarrow 0$,
$\alpha_m\rightarrow 1$ and $\kappa_m\rightarrow \infty, \forall m$,
and the DQLC becomes equivalent to uncoded
transmission.

To cancel interference at the receiver, sequential decoding is used: First an estimate of source $1$ is made. This estimate is then
subtracted from the channel output to estimate source $2$, and so on.
In order to make the correct decision on the output from quantizer $m$,
one must take into account that the midpoint of each \emph{channel segment} changes with $\rho_x$ (like $d_1$ and $d_2$ shown in Fig.~\ref{fig:sqlc_scheme_channel}).
Consider source 1: the first order moment of $p(\alpha_2 x_2,\ldots,\alpha_M x_M|x_1=q_{i_1})$ must be determined. Consider a
sub-division of $\mathbf{y}=[x_1\hspace{0.2cm}\cdots
\hspace{0.2cm}\alpha_M x_M]$ into $\mathbf{y}_a=[\alpha_2 x_2\hspace{0.2cm}\cdots \hspace{0.2cm}
\alpha_M x_M]^T$ and $\mathbf{y}_b= x_1$. By sub-dividing the
covariance matrix $K_\mathbf{y}=E\{\mathbf{y}\mathbf{y}^T\}$
according to~(\ref{e:subdiv_matrix}), Theorem~\ref{th:multvar_cond_moments} in the Appendix gives
$\EE \{\alpha_2 x_2,\ldots,\alpha_M x_M|x_1\}=\rho_x x_1
\begin{bmatrix} \alpha_2 &\cdots &\alpha_M \end{bmatrix}^T$. When
the transformed sources are summed together, centroids of
encoder 1 are shifted by the sum of all first order moments. When source 1 is
detected, one can subtract it from $z$ then use the same argument as above
for source 2, and so on. Let $N_{q_m}$ denote the number of centroids for quantizer $m$.
The estimate of sources $m\in[1,M-1]$ becomes
\begin{equation}\label{e:decode_two}
\begin{split}
&g_m(z)=\arg\min_{q_{i_m},i_m\in [1,N_{q_m}]}\\
&\bigg\|\bigg(z-\sum_{l=1}^{m-1}\alpha_l g_l(z)\bigg)-q_{i_m}\cdot\bigg(\alpha_m+\rho_x\sum_{k=m+1}^M \alpha_k\bigg)\bigg\|^2.
\end{split}
\end{equation}
The estimate of the $M$'th source is then found from
\begin{equation}
g_M(z)=\beta\bigg(z-\sum_{m=1}^{M-1}\alpha_m g_m(z)\bigg),
\end{equation}
where $\beta$ is a scaling factor. The MSE is further minimized by computing
\begin{equation}
\hat{x}_m = [x_m|g(x_1),\cdots,g(x_M)],\hspace{2mm} m\in[1,\cdots,M]
\end{equation}

The optimal parameters
$\alpha_m$, $\Delta_m$ and $\kappa_m$ need to be found. To ensure that each source is uniquely decodable, the relation between the $\alpha$'s must be $1=\alpha_1\geq\alpha_2\geq\alpha_3\geq\cdots\geq\alpha_M > 0$.
When optimizing DQLC, we look at the number of centroids of encoder $m$, $N_{q_m}$, instead of
the clipping $\kappa_m$ (except for encoder $M$). The parameters
$\overrightarrow{\Delta}=[\Delta_1,\cdots,\Delta_{M-1}]^T$,
$\overrightarrow{N_{q}}=[N_{q_2},\cdots,N_{q_{M-1}}]^T$,
$\overrightarrow{\alpha}=[\alpha_2,\cdots,\alpha_M]^T$, $\beta$ and
$\kappa_M$ are optimized for an average power constraint $P$.
With $D$ as in~(\ref{e:tot_dist})
\begin{equation}\label{e:opt_problem_sqlc}
\min_{\overrightarrow{\Delta},\overrightarrow{N_{q}},\overrightarrow{\alpha},\beta,\kappa_M : \sum_{m=1}^M P_m \leq M P} D.
\end{equation}
To calculate $D$ and $P_m$, the channel output pdf is needed.
\subsection{Calculation of channel output pdf}\label{ssec:pdf_calc_gen}
To determine the relevant pdf, $\kappa_m$, $\Delta_m$ and $\alpha_m$ must be chosen so that channel segments do not overlap (the configuration in Fig.~\ref{fig:sqlc_scheme_channel}).
Since the outputs of encoders 1 to $M-1$ are discrete, their
distributions can be expressed by point probabilities, which are straight forward
to calculate. For the output of encoder $M$, two
cases must be considered: 1) $\rho_x$ close enough to 0 for $\ell_{\pm \kappa_M}[x_M]$ to be
significant. 2) $\rho_x$ so close to $1$ that $p_\mathbf{x}(x_1,\cdots,x_M)$ effectively limits the segments so that $y_M=f_M(x_M)\approx \alpha_M x_M$.

Case 1): The whole range of $y_M=f_M(x_M)$ is now represented on each
channel segment, as can be seen by studying the blue lines in Fig.~\ref{fig:sqlc_scheme_source_r0}
and~\ref{fig:sqlc_scheme_channel}. The pdf of encoder $M$ at the channel output is determined by assuming that sources $1$ to $M-1$ are subtracted. Then the same analysis as for the $M=2$ case in~\cite{Floor_Kim_Wernersson11_TCOM} can be applied. With the mean given by
\begin{equation}\label{e:sum_pdf_mean}
\mu = q_{i_{M-1}}(\alpha_{M-1} +\rho_x \alpha_M)+ \sum_{m=1}^{M-2} q_{i_{m}},
\end{equation}
the same arguments as in~\cite{Floor_Kim_Wernersson11_TCOM} lead to
\begin{equation}\label{e:pdf_zr0}
\begin{split}
&p_{z_M (\overrightarrow{i}_m)}(z_M)_{\kappa_M} =\frac{1}{\Sigma}\int_{-\alpha_M\kappa_M}^{\alpha_M\kappa_M} e^{-\frac{\alpha_M^2\sigma_x^2 (z_M-\mu-y)^2+\sigma_n^2 y^2}{2\alpha_M^2\sigma_x^2\sigma_n^2}}\mbox{d}y\\ &+p_o\big(p_n(z_M-\mu-\alpha_M\kappa_M)+p_n(z_M-\mu+\alpha_M\kappa_M)\big),
\end{split}
\end{equation}
where $p_o = Pr\{x_M \geq \kappa_M\}$, $\Sigma=2\pi\alpha_M\sigma_x\sigma_n$, $p_n$ is the noise pdf, and $z_M$ denotes the received signal when sources $1\ldots M-1$ are subtracted. Fig.~\ref{fig:Calc_pdf} shows this pdf when $\mu=0$.

Case 2): When $\rho_x$  is close to 1, each source segment, and therefore also each channel segment, will no longer be equivalent but contain somewhat different (but intersecting) \emph{ranges} of $y_M$ (this also applies to sources 2 to $M-1$ given the others). This can be seen by studying the blue lines in Fig.~\ref{fig:sqlc_scheme_source_r095} and Fig.~\ref{fig:sqlc_scheme_channel}. One can now assume that $y_M=f_M(x_M)\approx \alpha_M x_M$. To determine the relevant pdf, $p(y_M|q_{i_1},\ldots,\alpha_{M-1} q_{i_{M-1}})$, after summation over GMAC must be found. Consider a
sub-division of $\mathbf{y}=[x_1\hspace{0.2cm}\cdots
\hspace{0.2cm}\alpha_M x_M]$ into $\mathbf{y}_a=\alpha_M x_M$
and $\mathbf{y}_b= [x_1 \hspace{0.2cm} \alpha_2 x_2
\hspace{0.2cm}\cdots \hspace{0.2cm}\alpha_{M-1}x_{M-1}]^T$. By
sub-dividing the covariance matrix $K_\mathbf{y}=\EE\{\mathbf{y}\mathbf{y}^T\}$ according
to~(\ref{e:subdiv_matrix}), Theorem~\ref{th:multvar_cond_moments} in the Appendix gives the second order moment
\begin{equation}\label{e:moments_pdf}
\begin{split}
&\Sigma_{aa\cdot b}=\sigma_x^2\alpha_M^2 \bigg(1-\frac{(M-1)\rho_x^2}{1+(M-2)\rho_x}\bigg).
\end{split}
\end{equation}
By inserting~(\ref{e:moments_pdf}) into~(\ref{e:gen_cond_pdf}) (in the Appendix) and~(\ref{e:sum_pdf_mean}) for the mean, the
wanted pdf results. After addition of noise, the resulting pdf
is found by the convolution $p(y_M|q_{i_1},\ldots,\alpha_{M-1} q_{i_{M-1}})\ast p_n(n)$~\cite[181-182]{papoulis02}, and thus
\begin{equation}\label{e:pdf_z3_rho}
p_{z_M(\overrightarrow{i}_m)}(z_M)_{\gamma_M}= \frac{1}{\sqrt{2\pi\big(\sigma_n^2+\Sigma_{aa\cdot b}\big)}}{e^{-\frac{1}{2}\frac{(\mu-z_M)^2}{\sigma_n^2+\Sigma_{aa\cdot b}}}}.
\end{equation}

The validity of~(\ref{e:pdf_zr0}) and~(\ref{e:pdf_z3_rho}) must be determined. Since the correlation between any two sources is assumed to be the same, one can focus on the $x_M,x_{M-1}$ plane. Fig.~\ref{fig:SQLC_corr_concept} provides a geometrical picture for the following discussion.
\begin{figure}[h]
    \begin{center}
         \includegraphics[width=1\columnwidth]{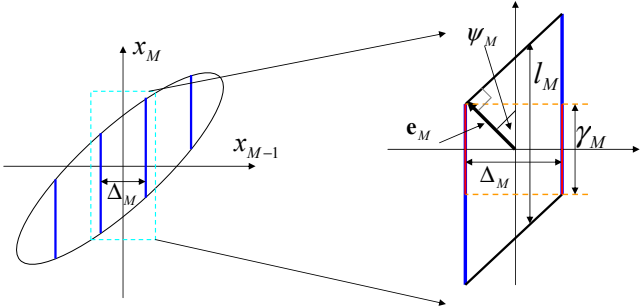}
    \end{center}
    \caption{DQLC seen in the $x_{M-1},x_M$-plane when $\rho_x= 0.95$. The expanded rectangle on the right can be applied to calculate anomalous errors for source $M$.}
     \label{fig:SQLC_corr_concept}
\end{figure}
Let $l_M$ denote the length of the portion of the $x_M$ axis that contains the \emph{significant probability mass}\footnote{``Significant probability mass'' means all events except those with very low probability.} given $x_{M-1}$ (or $q_{i_{M-1}}$).  $l_M=2\sqrt{\vartheta}\|\mathbf{e}_M\|=2 b_M\sqrt{\vartheta\lambda_M}$, where $\|\mathbf{e}_M\|=b_M\sqrt{\lambda_M}$ denotes the length of the minor axis of the ellipse depicted (the source space). $b_M$ ($\approx 4$) is a parameter determining the width of the ellipse shown, and should be chosen so that the significant probability mass is within this ellipse. $\vartheta=(l_M/(2\|\mathbf{e}_M\|))^2$ depends on $\rho_x$: If $\rho_x=0$, then $\vartheta = 1$ since the source space is rotationally invariant (a sphere). If $\rho_x > \approx 0.7$ then $\vartheta\approx 1/\cos^2 (\psi_M)=1/\cos^2 (\pi/4)=2$. That is, $\vartheta\in[1,2]$, depending on $\rho_x$. (\ref{e:pdf_zr0}) is valid when $l_M > 2\kappa_M$
while~(\ref{e:pdf_z3_rho}) is valid when $l_M \leq 2\kappa_M$.

The total channel output pdf is given by
\begin{equation}\label{e:GMAC_pdf_tot}
\begin{split}
p_z(z)=\sum_{\overrightarrow{i}_m} Pr\{\overrightarrow{i}_m\}  p_{z_M | \overrightarrow{i}_m}.
\end{split}
\end{equation}
$p_{z_M |\overrightarrow{i}_m}$ is given by ~(\ref{e:pdf_z3_rho}) or~(\ref{e:pdf_zr0}) depending on whether $l_1 > 2\kappa$ or not.
$\overrightarrow{i}_m = [i_1,i_2,\cdots,i_{M-1}]$, and $Pr\{\overrightarrow{i}_m\}= Pr\{q_{i_1}\} Pr\{q_{i_2}|q_{i_1}\} \cdots Pr\{q_{i_{M-1}}|q_{i_1},\cdots, q_{i_{M-2}}\}$.
Fig.~\ref{fig:Channel_output_pdf} shows an example of the channel output pdf when $M=3$ at 30 dB channel SNR.
\begin{figure}[h]
    \begin{center}
        \subfigure[]{
            \includegraphics[width=0.45\columnwidth]{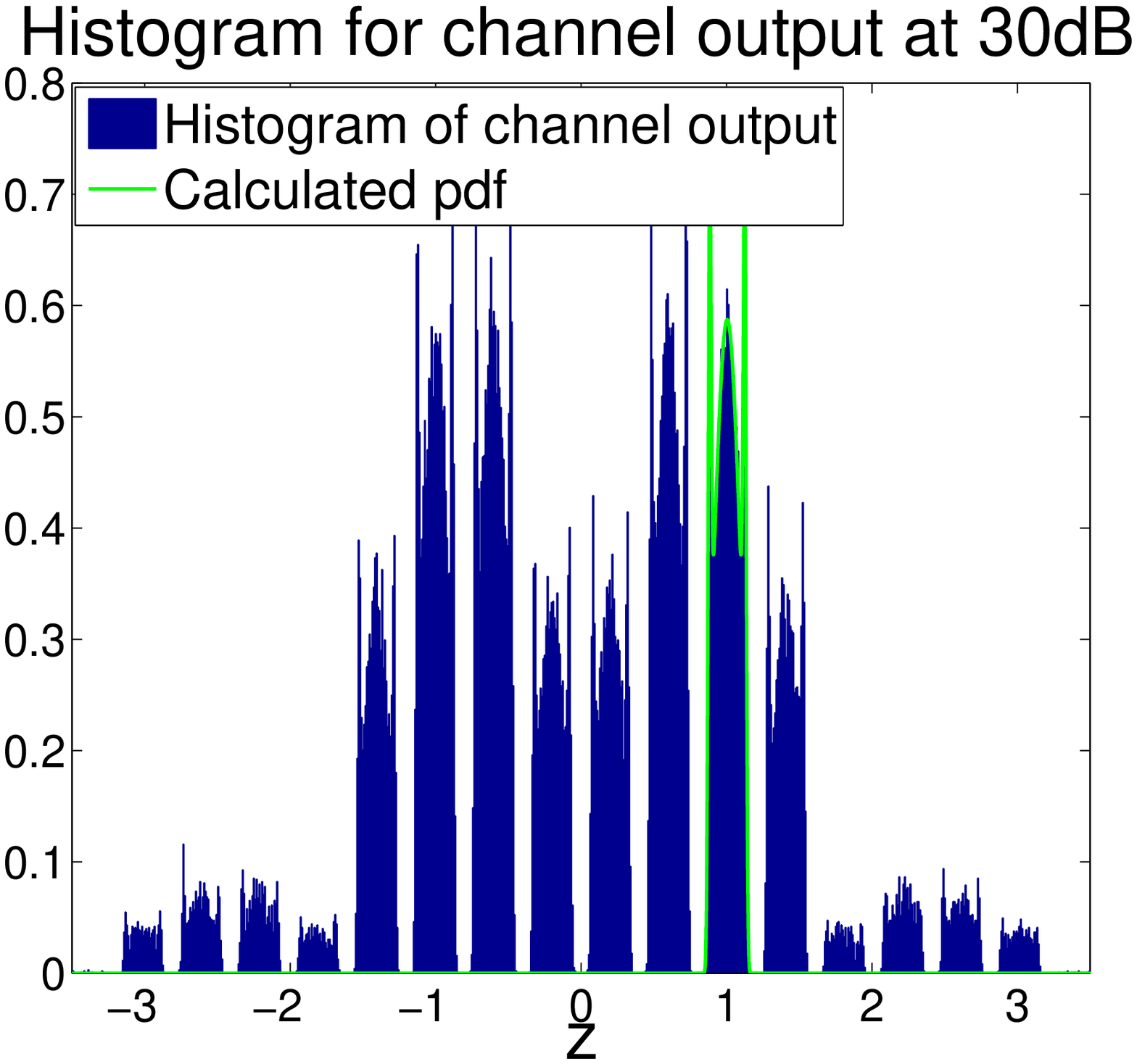}
        \label{fig:Histogram}}
        \hfil
        \subfigure[]{
            \includegraphics[width=0.45\columnwidth]{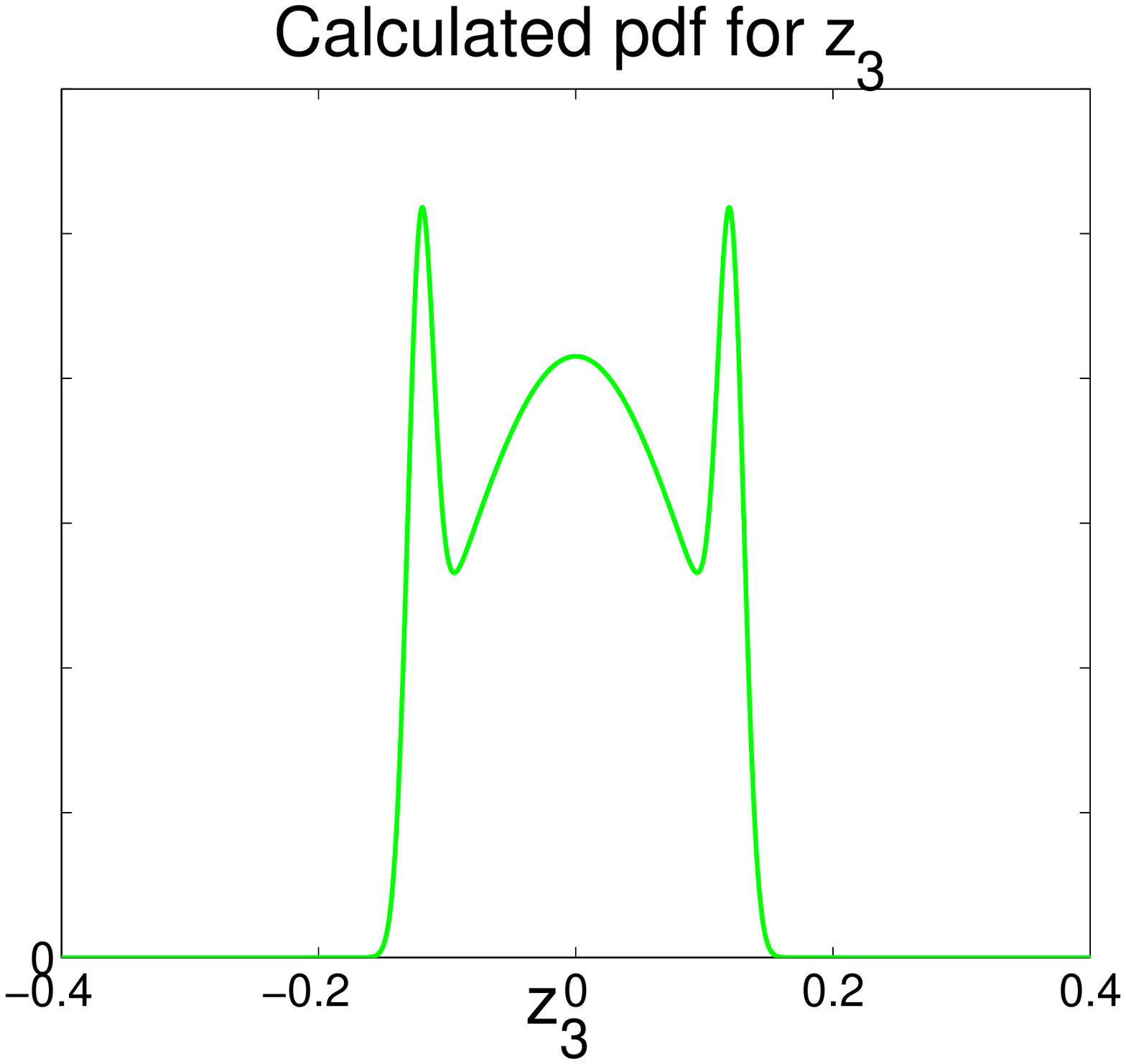}
        \label{fig:Calc_pdf}}
    \end{center}
    \caption{Channel output pdf when $M=3$. (a) Histogram where 4 centroids of encoder 1 is in use and $N_{q_{2}}=4$. The green curve shows~(\ref{e:pdf_zr0}) when the first positive centroid (index 1) are transmitted for both encoder 1 and 2. (b)~(\ref{e:pdf_zr0}) shown centered at the origin ($\mu=0$), that is when centroids for encoders 1 and 2 are given.}
    \label{fig:Channel_output_pdf}
\end{figure}
\subsection{Distortion and Power calculation}\label{ssec:dist_calcM}
To calculate the distortion, we use an approach similar to
that in~\cite{Floor_Kim_Wernersson11_TCOM}, where the total
distortion is divided into several contributions.

\subsubsection{Distortion and power for source $M$}\label{ssec:dist_M}
The distortion for source $M$ can be divided into three contributions: clipping distortion, anomalous distortion and channel distortion.

\emph{Clipping distortion}:
At encoder $M$ we only have
distortion from limitation, $\bar{\varepsilon}_{\kappa_M}^2$, an event with probability
$Pr\{|x_M|>\kappa_M\}$ and error $(x_M-\kappa_M)^2$:
\begin{equation}\label{e:cutoff_noise_M}
\bar{\varepsilon}_{\kappa{_M}}^2=2\int_{\kappa_M}^\infty (x_M-\kappa_M)^2 p_x(x_M)\mbox{d}x_M. 
\end{equation}

The distortion from channel noise can be split into two
contributions: \emph{Channel distortion} and \emph{anomalous
distortion}.

\emph{Anomalous distortion:} results from a
\emph{threshold effect} (see e.g.~\cite{shann49} or~\cite{Merhav_N_it_2011}) and occurs every time the centroid for one or several of quantizers 1 to $M-1$ is erroneously selected. This error leads to a "jump" from one channel
segment to another (see Fig.~\ref{fig:sqlc_scheme_channel}) for source $M$, resulting in large decoding errors. In the worst case scenario ($\rho_x=0$) large positive and negative values are interchanged.
The anomalous distortion is difficult to calculate exactly. We therefore look at an approximation valid around the optimal operation of DQLC. Note that jumps among centroids of in encoder 1 are
most fatal since this leads to anomalous distortion for all other sources. Jumps among centroids of encoder 2 lead to anomalous errors for sources 3 to $M$, and so on. Therefore, the probability for jumps among centroids of quantizer 1 to $M-2$ should be at least as small as the
probability for jumps among centroids of quantizer $M-1$. By assuming that encoders 1 to $M-1$ are constructed correctly, one can calculate an upper bound on anomalous distortion for source $M$ by considering jumps among centroids of quantizer $M-1$ only. We are then in the same situation
as the $M=2$ case in~\cite{Floor_Kim_Wernersson11_TCOM}, and can calculate the anomalous errors from the $x_{M-1},x_M$ plane shown in Fig.~\ref{fig:SQLC_corr_concept}. Two cases must be considered: $l_M > 2\kappa_M$ and $l_M \leq 2\kappa_M$.

Assume first $l_M > 2\kappa_M$:
anomalous errors happen whenever $y_M+n \geq d_{M-1}/2=\Delta_{M-1}(\alpha_{M-1}+\alpha_M\rho_x)/2$ (see $M=3$ case in Fig.~\ref{fig:SNQLC_corr_concept}), i.e. the probability for anomalies are
\begin{equation}\label{e:th3}
\begin{split}
p_{thM}&=
2 \int_{\frac{\Delta_{M-1}}{2}(\alpha_{M-1}+\alpha_M\rho_x)}^\infty p_{z_M (\overrightarrow{0}_m)}(z_M)_\kappa\mbox{d}z_M,
\end{split}
\end{equation}
where $p_{z_M}(z_M)_\kappa$ is given in~(\ref{e:pdf_zr0}). Since different values of $\overrightarrow{i}_m$ basically shifts $p_{z_M}(z_M)_\kappa$, the relevant probability can be calculated by setting $\mu=0$ in~(\ref{e:pdf_zr0}). $p_{thM}$ must be found numerically since the integral in~(\ref{e:th3}) has no closed form solution. The anomalous error's magnitudes are the same regardless of which segment we are at, and are bounded by $(2\kappa_M)^2$, since $\kappa_M$ gets interchanged with $-\kappa_M$ when channel segments start to overlap.

Now assume $l_M \leq 2\kappa_M$: the probability for this event is given by
\begin{equation}\label{e:th3r}
\tilde{p}_{thM}=2\int_{\frac{\Delta_{M-1}}{2}(\alpha_{M-1}+\alpha_M\rho_x)}^\infty p_{z_M(\overrightarrow{0}_m)}(z_M)_{\gamma_M}\mbox{d}z_M,
\end{equation}
where $p_{z_M(\overrightarrow{0})}(z_M)_{\gamma_M}$ is given by~(\ref{e:pdf_z3_rho}), where one again can assume that $\mu=0$.
When $\rho_x$ gets close to one, the anomalous errors, $\gamma_M$, become smaller. This can be seen in Figs.~\ref{fig:SNQLC_corr_concept} and~\ref{fig:SQLC_corr_concept}. Since $\gamma_M$ is approximately the same in magnitude regardless of which channel segment we jump from (see Fig.~\ref{fig:SQLC_corr_concept}), it can be calculated by considering jumps between the segments closest to the origin in the source space. The parallelogram shown on the right hand side of Fig.~\ref{fig:SQLC_corr_concept} applies to approximate $\gamma_M$.  Since $\psi_M=\pi/4$,  the parallelogram consists of a square
and two right triangles with both catheti equal to $\Delta_M$. This further implies that $\gamma_M\approx l_M-\Delta=2 b_M\sigma_x \sqrt{\vartheta(1-\rho_x)}-\Delta_M$,
where $\vartheta\approx2$ since $\rho_x$ is large (see Section~\ref{ssec:pdf_calc_gen}).

The anomalous distortion becomes
\begin{equation}\label{e:anM_noise_gen}
\bar{\varepsilon}_{anM}^2=
\begin{cases}
4 p_{thM} \kappa_M^2, & l_M > 2\kappa_M, \vspace{0.3cm}\\ \tilde{p}_{thM}\gamma_M^2, & l_M \leq 2\kappa_M.
\end{cases}
\end{equation}

\emph{Channel distortion}: Let $\tilde{x}_M=\ell_{\pm \kappa_M}[x_M]$. With no threshold effect occurring
the noise is additive and given by
\begin{equation}\label{e:chM_noise_gen}
\begin{split}
\bar{\varepsilon}_{C_M}^2 &= 
E\{(\tilde{x}_M-(\alpha_M\tilde{x}_M+n)\beta)^2\}\\&\approx\sigma_x^2(1-\alpha_M\beta)^2+\beta^2\sigma_n^2.
\end{split}
\end{equation}
The last approximation is based on the assumption that $E\{\tilde{x}_M\}\approx\sigma_x^2$.

\emph{Power}: With $p_o=Pr \{x_M\geq \kappa_M\}$, the output power from encoder $M$ becomes
\begin{equation}\label{e:power_M}
P_M = \int_{-\alpha_M\kappa_M}^{\alpha_M\kappa_M}y_M^2 p_{y_M}(y_M)\mbox{d}y_M+2 p_o \alpha_M^2\kappa_M^2.
\end{equation}
%
%
\subsubsection{Distortion and power for source $1$ to $M-1$}\label{ssec:dist_1to_m1}  Here we have quantization- and limitation distortion from the encoding process and channel distortion from channel noise.

\emph{Quantization and limitation}:
These contributions are equivalent to granular- and overload
distortion from the quantization process, that is
\begin{equation}\label{e:q_noise_gen}
\begin{split}
&\bar{\varepsilon}_{q,m}^2
= 2\sum_{i=1}^{\frac{Nq_m}{2}-1} \int_{(i-1)\Delta}^{i\Delta} \bigg(x_m-(i-1)\Delta-\frac{\Delta}{2}\bigg)^2 p_x(x_m)\mbox{d}x_m\\
&+2\int_{\big(\frac{Nq_m}{2}-1\big)\Delta}^\infty \bigg(x_m - \bigg(\frac{Nq_m}{2}-1\bigg)\Delta-\frac{\Delta}{2}\bigg)^2 p_x(x_m)\mbox{d}x_m.
\end{split}
\end{equation}

\emph{Channel distortion}:
The distortion from channel noise is given by 
\begin{equation}\label{e:nq_ndist_gen}
\begin{split}
  &\bar{\varepsilon}_{ch_t,m} =\overbrace{\int\cdots \int}^{M-\text{fold}} \int p_{\mathbf{x}}(x_1,\ldots,x_M)p\big(z | y_1,\ldots,y_M\big)\\
  &\big[\tilde{y}_m - \hat{x}_m(g_1(z),\ldots,g_M(z))\big]^2 \mbox{d}z \mbox{d}x_1\cdots \mbox{d}x_M,
\end{split}
\end{equation}
where $y_m=f_m(x_m)$, and $\tilde{y}_m$ denotes the quantized and
limited $x_m$. The relevant pdf can be derived
from~(\ref{e:GMAC_pdf_tot}).

As for the $M$'th
source, the distortion can be divided into channel distortion and
anomalous distortion, where channel distortion refers to jumps
among neighboring centroids and anomalous distortion refers
to the situation where large errors occur due to jumps from one
quantized ``segment'' to another. Take $M=3$: anomalies occur for $x_2$ and $x_3$ when the channel
noise takes us across the decision border for encoder 1 in
Fig.~\ref{fig:SNQLC_corr_concept} (``segment'' now refers to the collection of purple dots between each decision border of encoder 1). Anomalies do not
occur for source 1, i.e. when a centroid for encoder one is erroneously detected. For general $M$, anomalies occur for sources
$x_{m+1},\ldots,x_M$ when there is a channel error for source
$x_m$. The anomalous errors
for source $2,\ldots,M-1$ can be derived in a similar way as for
source $M$, as  illustrated in Section~\ref{sec:Ex_SNQLC_M3}. Channel
distortion for source $m$ is proportional to $\Delta_m^2$ and its probability can be determined from the probability for anomalous errors for source $m+1$
(e.g. the probabilities for channel distortion for $x_{M-1}$ can be calculated using~(\ref{e:th3}),~(\ref{e:th3r}), as illustrated in Section~\ref{sec:Ex_SNQLC_M3})

\emph{Power}:
The power from encoder $1$ to $M-1$ is given by\footnote{Note that~(\ref{e:q_noise_gen}) and~(\ref{e:power_i}) is derived for midrise quantizers. A similar expression can be derived for midthread quantizers.}
\begin{equation}\label{e:power_i}
P_m = 2\alpha_m^2\sum_{i=1}^{N_{q_m}/2}p_i\bigg((i-1)\Delta_m+\frac{\Delta_m}{2}\bigg)^2 
\end{equation}
where $p_i=Pr\{(i-1)\Delta_m < x_m \leq i \Delta_m\}$.

\subsection{High SNR analysis}\label{ssec:high_snr_dqlc}
From the previous section it is clear that a closed form expressions describing the DQLC in general is hard, if at all possible to find. One can, however, find closed form expressions that approximate the distortion well at high SNR.  These expressions can further be used to determine how well DQLC performs at high SNR as a function of both $M$ and $\rho_x$.

The performance of DQLC is compared to \emph{Performance upper bound}, i.e. the signal-to-distortion ratio SDR=$\sigma_x^2/D$. Assuming that SNR$\rightarrow\infty$, the bound ~(\ref{e:Dist_bound_equalP}), for the case SNR=$P/\sigma_n^2 >\rho_x/(1-\rho_x^2)$, becomes
\begin{equation}
\text{SDR}\approx \sqrt[M]{\frac{M\text{SNR}}{\big(1-\rho_x\big)^{M-1}}}=\varrho_M.
\end{equation}
The symbol, $\varrho_M$, is introduced to have a compact representation for later derivations.
Due to the fact that the code word length is short, there is a significant variance around the mean length of any stochastic vector~\cite[p. 324]{wozandj65} (for a normalized i.i.d. Gaussian vector $\bar{\mathbf{x}}$ of dimension $N$, $\text{Var}\{\|\bar{\mathbf{x}}\|\}=2\sigma_x^4/N$), making exact analysis difficult. To obtain closed form expressions, only the distortion terms that are dominant at high SNR are taken into account. It is further assumed that $\sigma_x=1$. Take source $M-1$: channel errors for this source and anomalous errors for source $M$ can (nearly) be avoided by assuming a distance $\Delta_{M-1} > 2\sqrt{(\alpha_M l_M/2)^2 + (b_n \sigma_n)^2}$ between each centroid (the purple dots in Fig.~\ref{fig:sqlc_scheme_channel}). $b_n=b_M$ ($\approx 4$) is a constant that must be chosen so that the significant probability mass of the noise is within $2 b_n \sigma_n$. A similar argument can be used for the other encoders. To quantify the magnitude of $\Delta_m$, Fig.~\ref{fig:DQLC_corr_concept_x2}, depicting the $x_1,x_2$ plane for the $M=3$ case after quantization, applies.
\begin{figure}[h]
    \begin{center}
     \includegraphics[width=0.9\columnwidth]{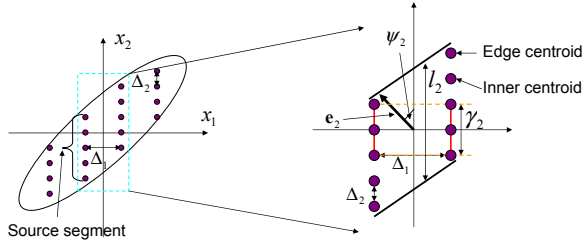}
    \end{center}
    \caption{$x_1,x_2$ plane of DQLC with $M=3$ when $\rho_x= 0.95$. The expanded rectangle on the right can be applied to calculate anomalous errors for source $2$.}
     \label{fig:DQLC_corr_concept_x2}
\end{figure}
From Fig.~\ref{fig:SNQLC_corr_concept} and~\ref{fig:DQLC_corr_concept_x2} one can convince oneself that if $\Delta_{m} > 2\alpha_{m+1} l_{m+1}/2$, the integrals in~(\ref{e:nq_ndist_gen}) can be avoided, since no distortion results from channel noise. For example, if $\Delta_1> 2\alpha_{2} l_{2}/2$ in Fig.~\ref{fig:sqlc_scheme_channel}, the green stars will not be confused. One can derive from Fig.~\ref{fig:DQLC_corr_concept_x2} that $l_{m+1}=2 b_{m+1} \sqrt{\vartheta_{m+1} (1-\rho_x)}$ (as was done for $l_M$ in section~\ref{ssec:pdf_calc_gen}). For large SNR, $\kappa_m$ become so large that one may neglect the $\ell_{\pm \kappa_m}[x_m]$ operation. The high-rate approximation to a scalar
quantizer, $\Delta_m^2/12$, then applies to quantify distortion. To avoid constrained optimization, $D_1$ is further scaled by $P_1= M P - \sum_{i=2}^M P_i$. For convenience, we also scale $D_m$ by $P_m$ prior to quantization (instead of after). Note that when SNR gets large enough, $\kappa$ becomes so large and $\Delta_m$ becomes so small that $P_m\approx \alpha_m^2\sigma_x^2$.

The distortion at high SNR can be approximated by
\begin{equation}\label{e:dist_high_snr}
\begin{split}
&D_1 \approx \frac{\alpha_2^2 C_2(1-\rho_x)}{3\bigg(M P -\sum_{i=2}^M\alpha_i^2\bigg)}\\& D_m \approx \frac{\alpha_{m+1}^2 C_{m+1}(1-\rho_x)}{3\alpha_m^2 },\hspace{0.1cm} m\in[2,M-2]\\
&D_{M-1} \approx \frac{\alpha_M^2 C_M(1-\rho_x) + (b_n \sigma_n)^2}{3\alpha_{M-1}^2}, \hspace{0.3cm} D_M\approx\frac{\sigma_n^2}{\alpha_M^2},
\end{split}
\end{equation}
where $C_m = b_m^2\vartheta_m^2 $. For $D_M$, the high SNR approximation of~(\ref{e:chM_noise_gen}) was used. 

To determine the optimal performance, the optimal $\alpha_m^2$ needs to be found. The obvious way is to solve ${\nabla_{\overrightarrow{\alpha}} \big[D_1+D_2+\cdots+D_M\big]}=0$ with
respect to $\alpha_m^2$. But since these are equations of order $\geq 4$, analytical solutions can not be found, and a different approach must be chosen.
From the $M=2$ case in~\cite{Floor_Kim_Wernersson11_TCOM} is is known that the distortion of DQLC is a constant times the bound $\varrho_2$. We therefore
hypothesize that this is the case for general $M$ as well: By choosing $\alpha_M^2=\sigma_n^2 \mathcal{K} \varrho_M$,
then $D_M=\mathcal{K} \varrho_M$. Since we want $D_1=D_2=\cdots =D_M=1/(\mathcal{K}\varrho_M)$, the $\mathcal{K}$ that satisfies this relation must be determined.
By solving $1/D_{M-1} = \mathcal{K}\varrho_M$, then
\begin{equation}
\begin{split}
\alpha_{M-1}^2&=\frac{\sigma_n^2}{3} C_M (1-\rho_x) \big(\mathcal{K}\varrho_M\big)^2\\
&=\frac{\alpha_M^2}{3} C_M (1-\rho_x) \big(\mathcal{K}\varrho_M\big),
\end{split}
\end{equation}
assuming $b_n\sigma_n\approx 0$. By continuing to solve $1/D_{m}=\mathcal{K}\varrho_M$ with respect to $\alpha_m^2$ for $m$ from $M-2$ to $2$ using the previously derived $\alpha_{m+1}^2$, one can show that
\begin{equation}\label{e:optalph_high_SNR}
\begin{split}
\alpha_i^2&=
  \frac{\alpha_M^2}{3^{M-i}} \bigg(\prod_{j=i+1}^M C_j\bigg) (1-\rho_x)^{M-i} \big(\mathcal{K}\varrho_M\big)^{M-i}.
\end{split}
\end{equation}
Finally, the $\mathcal{K}$ that makes $1/D_1 = \mathcal{K}\varrho_M$, must be found. Expanding the sum $\sum_{i=2}^M \alpha_i^2$ using~(\ref{e:optalph_high_SNR}), and setting $C_{M+1}=1$, one can show that
\begin{equation}\label{e:alsum_pre}
\sum_{i=2}^M \alpha_i^2 =\alpha_M^2 \sum_{k=1}^{M-1} \bigg(\frac{(1-\rho_x)\mathcal{K}\varrho_M}{3}\bigg)^{k-1}\prod_{j=M+2-k}^M C_j.
\end{equation}
Letting $C_2=C_3=\cdots C_M= C$ (which is the case when $\rho_{ij}=\rho_x, \forall i,j$ at high SNR), the product in~(\ref{e:optalph_high_SNR}) becomes $C^{-1}C^{k-1}$, and~(\ref{e:alsum_pre}) turns into a Geometric series.
From the sum of a Geometric series
\begin{equation}\label{e:geom_series}
\sum_{k=1}^n q^{k-1}=\frac{1-q^n}{1-q},
\end{equation}
the following high SNR approximation can be derived
\begin{equation}\label{e:alpha_sum_highSNR}
\sum_{i=2}^M \alpha_i^2 \approx \frac{\alpha_m^2}{C}\bigg(\frac{(1-\rho_x) C \mathcal{K}\varrho_M}{3}\bigg)^{M-2}.
\end{equation}
Inserting~(\ref{e:alpha_sum_highSNR}) and~(\ref{e:optalph_high_SNR}) with $i=2$ into the expression for $D_1$ in~(\ref{e:dist_high_snr})
\begin{equation}\label{e:d1_solve}
\begin{split}
&\bigg(M\text{SNR}-\frac{\alpha_m^2}{C}\bigg(\frac{(1-\rho_x) C \mathcal{K}\varrho_M}{3}\bigg)^{M-2}\bigg)\cdots\\&\bigg(\frac{(1-\rho_x) C \mathcal{K}\varrho_M}{3}\bigg)^{1-M}=\mathcal{K}\varrho_M,
\end{split}
\end{equation}
where SNR$=P/\sigma_n^2$. By solving~(\ref{e:d1_solve}) and removing constant terms, then
\begin{equation}\label{e:SDR_loss_highSNR}
\mathcal{K}=\sqrt[M]{\bigg(\frac{3}{C}\bigg)^{M-1}}=\bigg(\frac{3}{C}\bigg)^{1-\frac{1}{M}}, \hspace{0.25cm}\rho_x\neq1.
\end{equation}
~(\ref{e:SDR_loss_highSNR}) quantifies the loss to the performance upper bound. By inserting $C=b^2 \vartheta$ and $M=2$ in~(\ref{e:SDR_loss_highSNR}), the loss calculated in~\cite{Floor_Kim_Wernersson11_TCOM} results, and so~(\ref{e:SDR_loss_highSNR}) is a generalization. One can observe that for any $M$ and $\rho_x$, DQLC exhibits a constant gap to the bound as SNR$\rightarrow\infty$.  The gap grows somewhat with $M$, however, but is fortunately bounded: Taking the limit $M\rightarrow\infty$, the loss is $\approx 7.2$ dB when $\rho_x$ is close to 0 and $\approx 10.2$ dB when $\rho_x$ is close to 1. From~(\ref{e:SDR_loss_highSNR}) one can see that there are mainly two loss factors. One is due to short code length: When the code length is infinite, $b\rightarrow 1$ (see e.g.~\cite{Floor_kim_ramstad_ITW2012}), whereas when the code length is one, $b\approx 4$. The nested structure of DQLC results in an increased loss with $M$ whenever $b>1$. This accumulation is avoided with the S-K mappings described in Section~\ref{ssec:colaborative_sk}, as seen from Fig.~\ref{fig:DQLC_High_SNR}. The distance to the upper bound is around 0.8-0.95 dB when $\rho_x = 0$, actually decreasing slightly when $M$ increases (the same effect is observed when $\rho_x$ is close to 1, but the loss is now around 2 dB). This clearly indicates that it is not the zero delay requirement alone that makes the loss increase with $M$. The reason for the increasing loss is most likely that the DQLC is sub-optimal. Alternatively, it may be that zero delay distributed JSCC schemes will suffer from an increasing loss like~(\ref{e:SDR_loss_highSNR}). This must be dismissed or confirmed through further research, however.
The reason why $\vartheta=2$ when $\rho_x$ is close to one is that the source space can not be rotated (decorrelation) with the choice of distributed encoders, and this implies
that each channel segment gets somewhat longer than necessary.

By inserting $b\approx 4$ and $\vartheta\approx 1$ when $\rho_x$ is close to zero, and $\vartheta \approx 2$ when $\rho_x$ is close to 1, the performance of DQLC as a function of $M$ can be plotted. Fig.~\ref{fig:Loss_at_high SNR} shows the loss from upper bound for DQLC (high SNR in general) and uncoded transmission at 100 dB SNR. S-K mappings for $\rho_x=0$ case is also shown (results are taken from~\cite[chapter 3]{fulds97a} and~\cite{FloorThesis}).
\begin{figure}[h]
    \begin{center}
        \subfigure[]{
            \includegraphics[width=1\columnwidth]{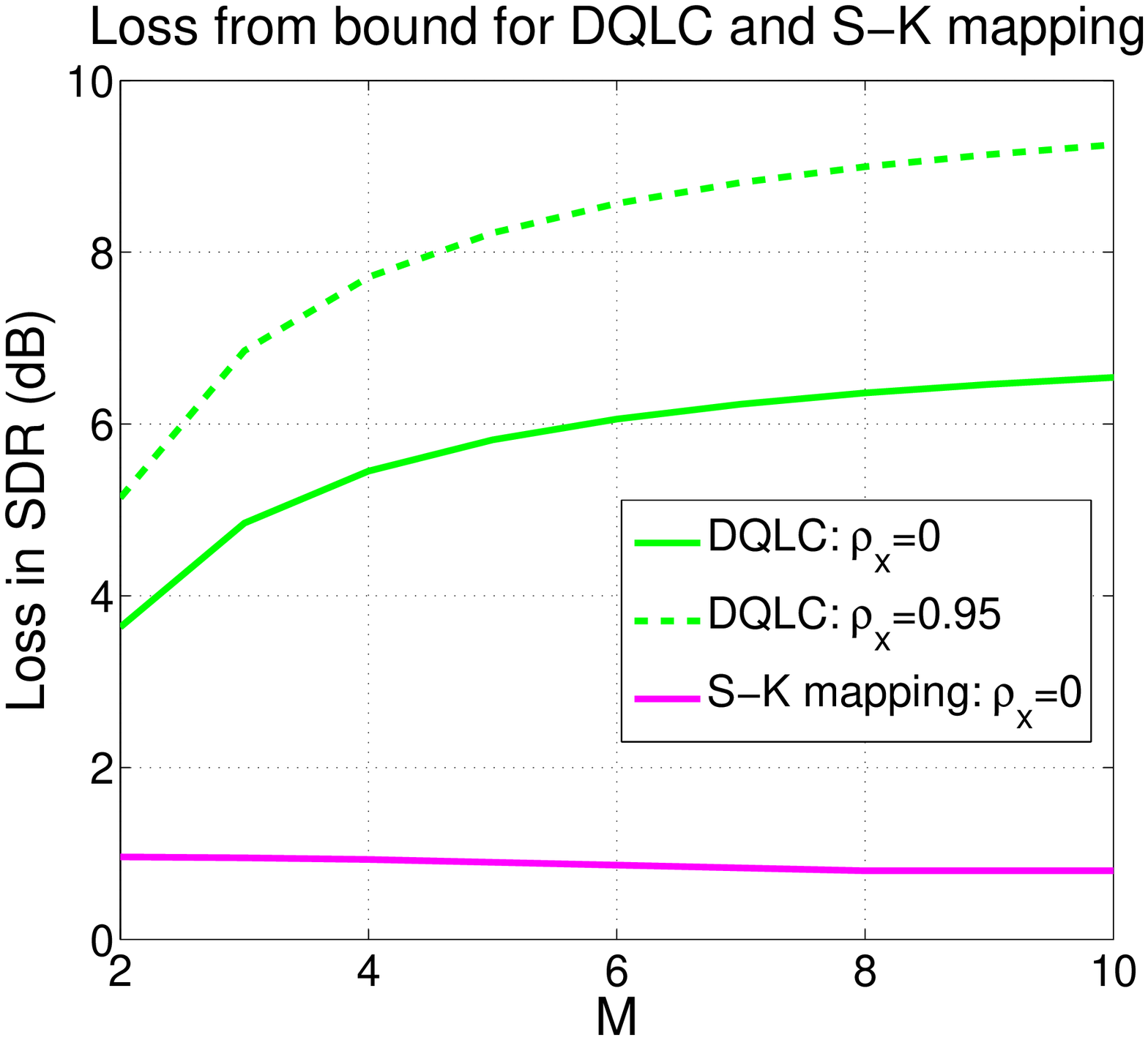}
        \label{fig:DQLC_High_SNR}}
        \hfil
        \subfigure[]{
            \includegraphics[width=1\columnwidth]{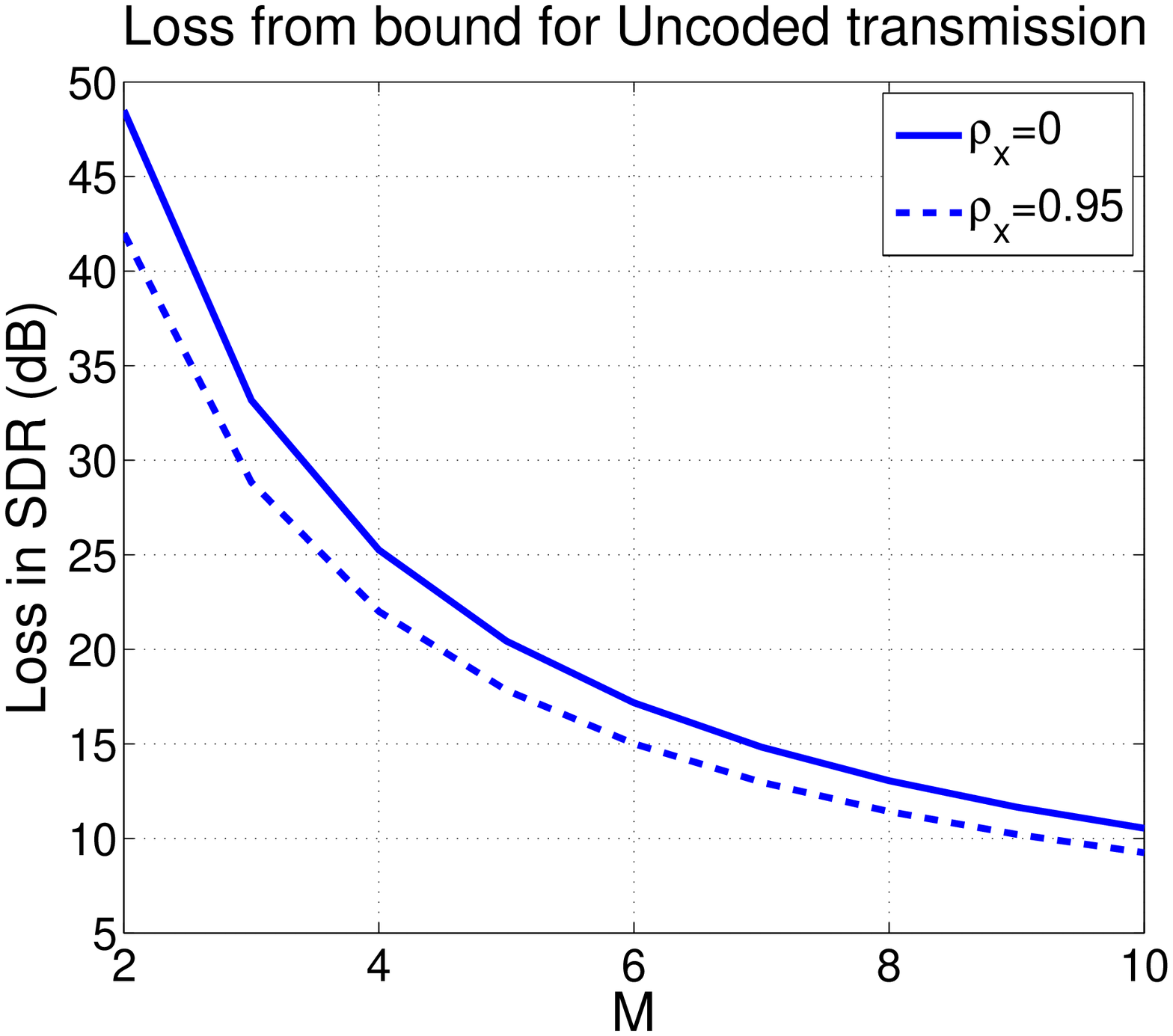}
        \label{fig:Uncoded_High_SNR}}
    \end{center}
    \caption{Loss in SDR from upper bound at SNR$=100$ dB. (a) DQLC. (b) Uncoded transmission.}
    \label{fig:Loss_at_high SNR}
\end{figure}
From these plots one can see that the performance of uncoded transmission will close in on the performance of DQLC when around 10 sources are considered.
This number will decrease somewhat for smaller SNR  and increase as the SNR gets higher. Note that for uncoded transmission the distance to the bound will in any case increase as the SNR grows (except when $\rho_x=1$), contrary to DQLC, and so DQLC will improve over uncoded transmission when the SNR gets sufficiently large, thus fulfilling our objective.

\section{Optimization and simulation of DQLC at arbitrary SNR when $M=3$}\label{sec:Ex_SNQLC_M3}
We give an example on how to calculate and optimize distortion for all SNR when $M=3$ by applying the analysis in
Section~\ref{ssec:dist_calcM}. The optimized DQLC is further simulated and compared to distortion lower bound, S-K mapping and uncoded transmission.

\subsection{Calculation of distortion}\label{ssec:calc_m3}
The distortion for source $3$ is found from Section~\ref{ssec:dist_calcM} by setting $M=3$ in~(\ref{e:cutoff_noise_M}),~(\ref{e:anM_noise_gen}) and~(\ref{e:chM_noise_gen}),
and  the channel power is given by~(\ref{e:power_M}). Furthermore, the distortion from quantization and limitation of source $1$ and $2$ is given by~(\ref{e:q_noise_gen}) and the power is given by~(\ref{e:power_i}). What remains to calculate is the distortion due to channel noise for source $1$ and $2$ given by~(\ref{e:nq_ndist_gen}). When analyzing DQLC around its optimal distortion, only jumps to the nearest neighboring centroids needs to be considered. This will simplify the calculations and speed up the optimization process. As for source 3, ~(\ref{e:nq_ndist_gen}) can be divided into channel- and anomalous distortion. As mentioned in Section~~\ref{ssec:dist_1to_m1}, both channel distortion and anomalous distortion  may occur for source 2, while only channel distortion may occur for source 1.

\emph{Anomalous distortion for source 2}:
A similar analysis to that in Section~\ref{ssec:dist_M} applies here. As for anomalous distortion for source 3, there are two cases to consider: $l_2> 2\kappa_2$ and $l_2\leq 2\kappa_2$. Now $l_2= b_2 \sigma_x \sqrt{2(1-\rho_x)}$ and calculated in the same way as $l_M$ in Section~\ref{ssec:pdf_calc_gen}, now using the parallelogram to the right in Fig.~\ref{fig:DQLC_corr_concept_x2}.
$b_2$ is a constant that determines the width of the ellipse (in Fig.~\ref{fig:DQLC_corr_concept_x2}) containing the significant (quantized) probability mass.

Assume first that  $l_2> 2\kappa_2$ (see $\rho_x=0$ case in Fig.~\ref{fig:SNQLC_corr_concept}). The error we get when anomalies start to
happen is the difference in magnitude between centroid no. 1 and centroid no. $N_{q_2}$, i.e. $|q_{({N_{q_2}})_2}-q_{1_2}|=\Delta_2(N_{q_2}-1)$.
The probability for anomalies (the probability for crossing of the decision borders for encoder 1 in Fig.~\ref{fig:sqlc_scheme_channel}) is given  by 
\begin{equation}
p_{th2}=Pr\{|f_2(x_2)+f_3(x_3)+n| \geq d_1/2\}=P_{N_{q_2}} p_{d_{th}},
\end{equation}
where $P_{N_{q_2}}=Pr\{(N_{q_2}-1)\Delta_2/2 < x_2 < \infty \}$ is the probability for being at an \emph{edge}
centroid (see Fig.~\ref{fig:sqlc_scheme_channel}) and $p_{d_{th}}=Pr\{|y_3+n|>d_{th}\}$, where
\begin{equation}\label{e:d_th}
d_{th}=\frac{d_1}{2}- \frac{\alpha_2 d_2}{2} (N_{q_2}-1),
\end{equation}
is the distance from an edge centroid of encoder 2 to the decision border of encoder 1 (see Fig.~\ref{fig:sqlc_scheme_channel}). The distribution needed to calculate $p_{d_{th}}$ is given in~(\ref{e:pdf_zr0}).

Now consider the case $l_2\leq 2\kappa_2$ ($\rho_x=0.95$
scenario in Fig.~\ref{fig:SNQLC_corr_concept}). This case is difficult to calculate as accurately as the $l_2> 2\kappa_2$ case above, since its hard to determine exactly which centroids that are edge centroids (this can be seen by comparing the $\rho_x=0.95$ case with the $\rho_x=0$ case in Fig.~\ref{fig:SNQLC_corr_concept}). One can get around this problem by calculating both the probability $\tilde{p}_{th2}=Pr\{|f_2(x_2)+f_3(x_3)+n| \geq d_1/2\}$
and the resulting anomalous error assuming that $y_2=f_2(x_2)$ is continuous. By using $p_{z_M(\overrightarrow{i}_m)}(z_M)_{\gamma_M}$ in~(\ref{e:pdf_z3_rho}) and the expression for $p(y_2|x_1)$ (see~\cite[p.223]{papoulis02}) one can show that
\begin{equation}\label{e:tpth2}
\begin{split}
&\tilde{p}_{th2}= \int_{d_1/2}^\infty \big[p_{z_M(\overrightarrow{i}_m)}(z_M)_{\gamma_M}\ast p(y_2|x_1) \big](z)\mbox{d}z\\
&=\int_{d_1/2}^\infty \frac{e^{-\frac{1}{2}\frac{(\mu-z_M)^2}{\sigma_n^2+\Sigma_{aa\cdot b}+\sigma_x^2\alpha_2^2 (1-\rho_x^2)}}}{\sqrt{2\pi\big(\sigma_n^2+\Sigma_{aa\cdot b}+\sigma_x\alpha_2 \sqrt{2\pi(1-\rho_x^2)}\big)}}\mbox{d}z,
\end{split}
\end{equation}
where $d_1=\Delta_1\big(1+\rho_x(\alpha_2+\alpha_3)\big)$.
As in Section~\ref{ssec:dist_M}, the probability in~(\ref{e:tpth2}) is calculated assuming $\mu=0$.
%
With the same reasoning as in Section~\ref{ssec:dist_M}, one can show that the error we get when anomalies start to happen is approximately
$\gamma_2=l_2-\Delta_1= b_2\sigma_x \sqrt{2(1-\rho_x)}-\Delta_1$ (see Fig.~\ref{fig:DQLC_corr_concept_x2}).

The anomalous distortion becomes
\begin{equation}\label{e:anomal_x2}
\bar{\varepsilon}_{an2}^2=
\begin{cases}
 p_{th2} \big(\Delta_2(N_{q_2}-1)\big)^2, & l_2 > 2\kappa_2, \vspace{0.3cm}\\
\tilde{p}_{th2}\gamma_2^2, & l_2 \leq 2\kappa_2,
\end{cases}
\end{equation}
where $2\kappa_2= \Delta_2(N_{q_2}-1)$.


\emph{Channel distortion, source 2}: Since anomalous errors results for source 3 whenever channel errors occur for source 2, the validity for these events are the same.
 A similar expression to that in~(\ref{e:anomal_x2}) can therefore be derived:

Consider first $l_3 > 2\kappa_3$: For a given source segment in the $x_1,x_2$ plane an \emph{inner centroid} has two neighbors, while an edge centroid  has only one (see Fig.~\ref{fig:DQLC_corr_concept_x2}).
The probability that an inner centroid is confused with its neighbors is given by $p_{th3}$, found by substituting $M=3$
in~(\ref{e:th3}). For edge centroids we have $p_{th3}/2$ since channel errors happens if an edge centroid is exchanged with an inner centroid, whereas anomalous errors
happen otherwise (see Fig.~\ref{fig:sqlc_scheme_channel}). When neighboring centroids are exchanged, the error is $\Delta_2^2$. With $P_{N_{q_2}}$,
the probability for an edge centroid, then 
\begin{equation}\label{e:distcalc_ch2}
\begin{split}
\bar{\varepsilon}_{C_2}^2&=2\Delta_2^2 p_{th_3} \sum_{i=2}^\frac{N_{q_2}}{2} P_i +\Delta_2^2 2 P_{N_{q_2}}\frac{p_{th_3}}{2}\\ &= 2\Delta_2^2 p_{th_3} \big(1/2-P_{N_{q_2}}\big) +\Delta_2^2 P_{N_{q_2}} p_{th_3}\\
 &= \Delta_2^2 p_{th_3} \big(1-P_{N_{q_2}}\big).
\end{split}
\end{equation}

Now assume $l_3 \leq 2\kappa_3$: As for anomalous distortion, this case is difficult to calculate accurately, due to the difficulty of determining which centroids are edge centroids.
One can, however, upper bound channel distortion by assuming that jumps from any centroid leads to channel distortion (i.e. none of the possible events leads to anomalous distortion).
The  probability for channel distortion is then given by $\tilde{p}_{th3}$, found by substituting $M=3$
in~(\ref{e:th3r}).  We therefore have $\bar{\varepsilon}_{C_2}^2\leq\Delta_2^2 \tilde{p}_{th_3}$. This bound is accurate enough to find the optimal parameters of DQLC.

\emph{Channel distortion, source 1}: Since $N_{q_1}=\infty$ and no threshold effects happen for source 1, the probability for channel distortion for source 1 will be the same as
the probability for threshold effects for source 2. The magnitude of the error is $\Delta_1^2$. Therefore
$\bar{\varepsilon}_{C_1}^2= \Delta_1^2 p_{th_2}$ when $l_2 > 2\kappa_2$, and
$\bar{\varepsilon}_{C_1}^2= \Delta_1^2 \tilde{p}_{th_2}$ when $l_2 \leq 2\kappa_2$.

\subsection{Optimization and Simulation}\label{ssec:sim_m3}
Instead of solving the constrained problem in~(\ref{e:opt_problem_sqlc}), we choose to scale $x_1$ by $\xi = \sqrt{3P-(P_2(\Delta_2,\alpha_2)+P_3(\alpha_3,\kappa_3))}/\sigma_x$ prior to quantization to get an unconstrained problem. Note that the factor $q_{i_m}(1+\rho_x(\alpha_2+\alpha_3))$ in~(\ref{e:decode_two}) must then be changed to $q_{i_m}(1+\rho_x(\alpha_2+\alpha_3)/\xi)$ in order to decode correctly (the integration limit in~(\ref{e:tpth2}) must also be changed).
The average distortion for DQLC is $D=(D_1+D_2+D_3)/3$, where
\begin{equation}
\begin{split}
D_1&=\frac{\bar{\varepsilon}_{q,1}^2(\Delta_1)+\bar{\varepsilon}_{C_1}^2(\Delta_1,\alpha_2,N_{q_2})}{\xi(\Delta_2,\alpha_2,\alpha_3,\kappa_3)},\\
D_2&=\bar{\varepsilon}_{q,2}^2(\Delta_2,N_{q_2})+\bar{\varepsilon}_{C_2}^2(\Delta_2,N_{q_2},\alpha_2,\alpha_3,\kappa_3)\\
&+\bar{\varepsilon}_{an2}^2(\Delta_2,N_{q_2},\alpha_2,\alpha_3,\kappa_3), \\
D_3&=\bar{\varepsilon}_{\kappa_3}^2(\kappa_3)+\bar{\varepsilon}_{C_3}^2(\alpha_3,\beta)+\bar{\varepsilon}_{an3}^2(\Delta_2,\alpha_2,\alpha_3,\kappa_3).
\end{split}
\end{equation}
All parameters must be greater than zero and $\alpha_2\geq\alpha_3$.
Some of the distortion terms do not have analytical solutions, and thus numerical optimization is necessary to determine the optimal parameters.

DQLC is compared to the bound from Section~\ref{ssec:Dist_bound}, uncoded
transmission from Section~\ref{sec:linear} and the S-K
mappings from Section~\ref{ssec:colaborative_sk}. Again we choose to look at signal-to-distortion ratio (SDR)
$\sigma_x^2/D$ as a function of channel SNR $=P/\sigma_n^2$ (and $\rho_x$) instead of distortion.
The results are shown in Fig.~\ref{fig:Results_M3_Rc} for $\rho_x=0$ and
$\rho_x=0.95$.
\begin{figure}[h!]
\centering
  \includegraphics[width=1\columnwidth]{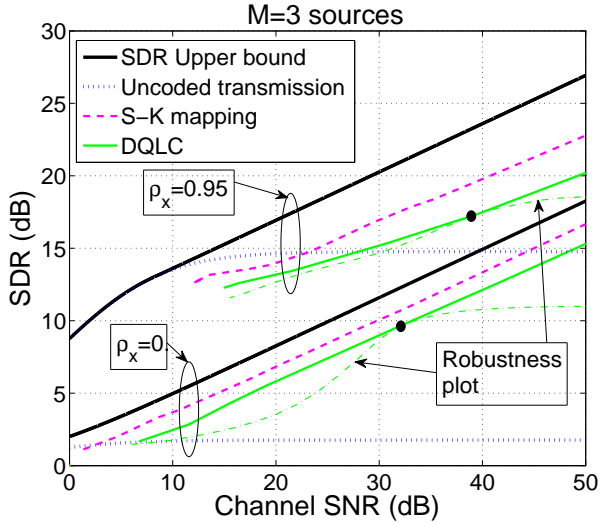}
  \caption{Performance of DQLC compared to the SDR upper bound, S-K mappings and uncoded transmission. The black dots shows the design SNR for the robustness plots for DQLC. }\label{fig:Results_M3_Rc}
\end{figure}
The SDR upper bound is given by $\sigma_x^2/D^*$, where $D^*$ is the optimal distortion given in~(\ref{e:Dist_bound_equalP}).

When $\rho_x=0$, DQLC drops around 2.5-3.5dB from the upper bound, while it drops around 4 to 7 dB when $\rho_x=0.95$. The loss at high SNR (50dB) corresponds well with the calculated estimate shown in Fig.~\ref{fig:DQLC_High_SNR} when $\rho_x=0.95$, while the calculated estimate is a bit pessimistic when $\rho_x = 0$. Fortunately, DQLC gets somewhat closer to the bound as the SNR drops. There is also a backoff to the
S-K mapping of around 1 to 2 dB. This is because S-K mappings avoid threshold effects as well as the ``loss accumulation''  mentioned in Section~\ref{ssec:high_snr_dqlc}. Interestingly, DQLC improves with increasing $\rho_x$ without changing the basic encoder and decoder
structure, only the different parameters needs to be adapted.
The improvement is significant, around 5 to 7 dB when $\rho_x$ goes from
0 to 0.95. DQLC is also robust against variations in noise level. Note that DQLC outperform uncoded transmission for SNR$>7$dB when $\rho_x=0$ and SNR$>28$dB when $\rho_x=0.95$.

The gain from increasing correlation as a
function of $\rho_x$ is shown in Fig.~\ref{fig:gain_rho_dqlc} for DQLC and uncoded transmission at 40 dB channel SNR.
\begin{figure}[h]
    \begin{center}
        \subfigure[]{
            \includegraphics[width=0.45\columnwidth]{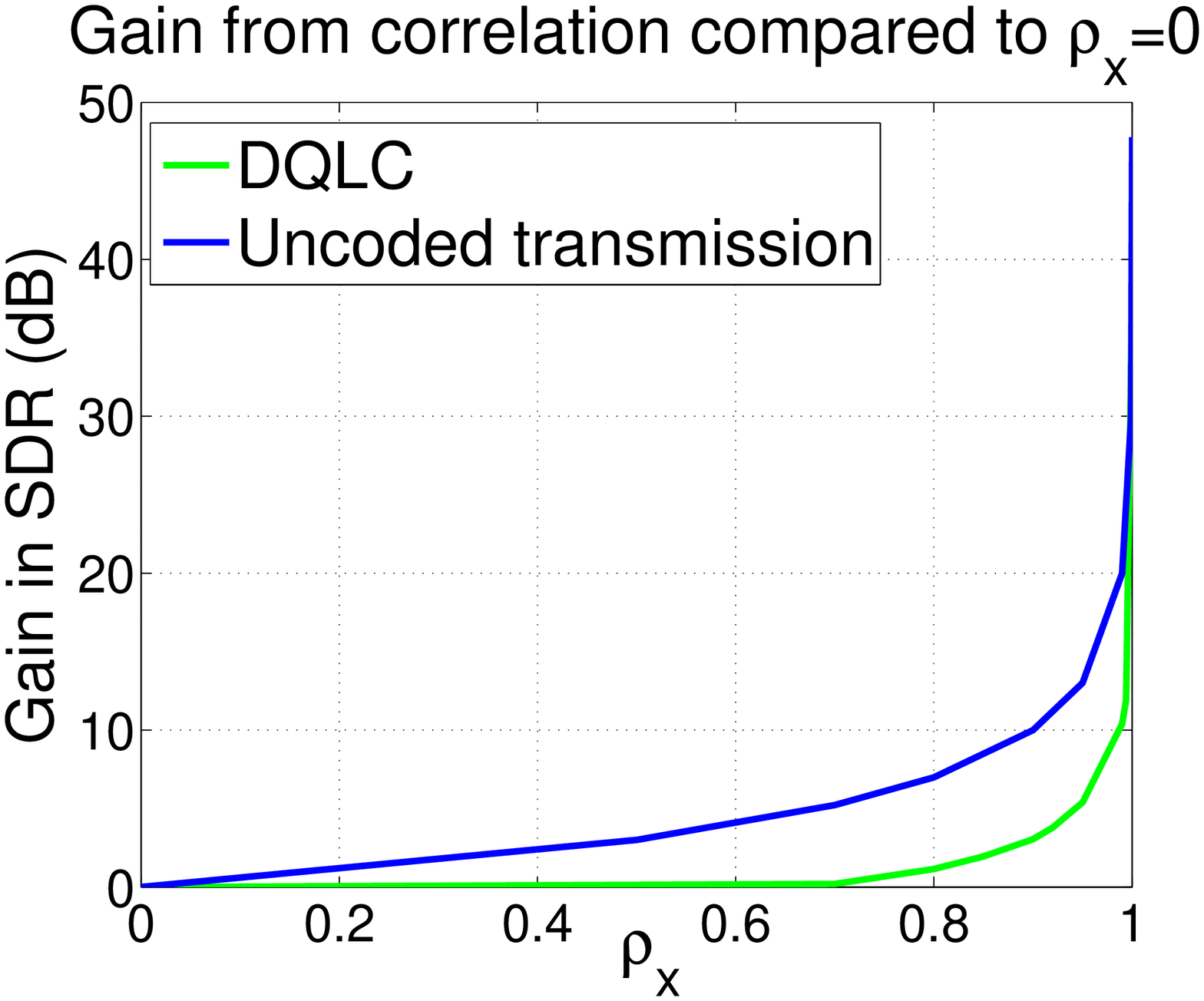}
        \label{fig:gain_rho_dqlc}}
        \hfil
        \subfigure[]{
            \includegraphics[width=0.45\columnwidth]{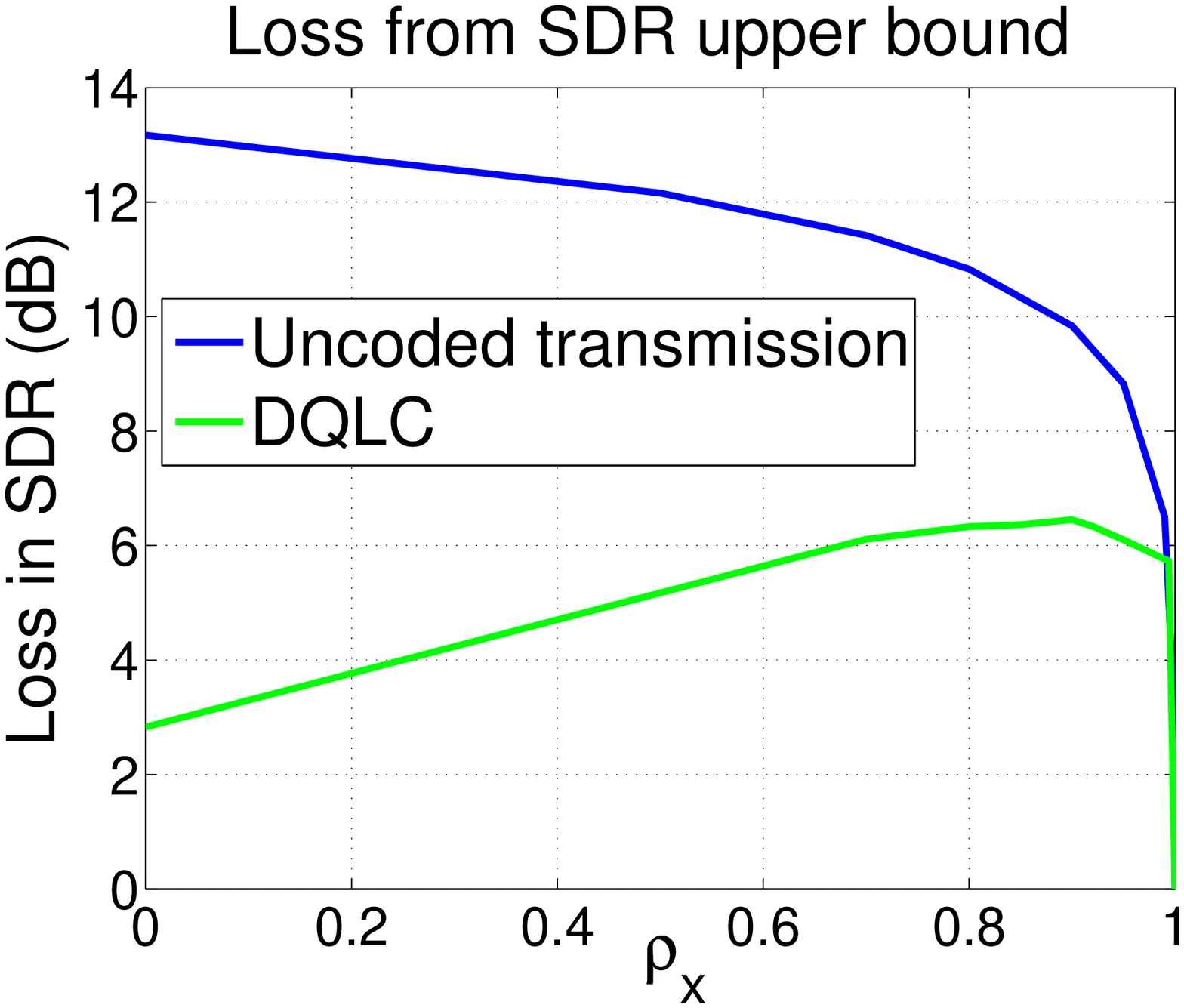}
        \label{fig:loss_opta_rho}}
    \end{center}
    \caption{How correlation affects performance when $M=3$ at 40 dB channel SNR. (a)
    Gain from correlation for DQLC and uncoded transmission. (b)
    Loss from SDR upper bound as a function of $\rho_x$.}\label{fig:Correlation_influence}
\end{figure}
Note that the gain for DQLC is not significant before $\rho_x > \approx
0.7$, whereas the gain gets large when $\rho_x\rightarrow 1$ (around 37.5 dB). Uncoded transmission shows an even greater gain, which is natural since it goes from being highly sub-optimal when $\rho_x=0$  to achieve the bound for all SNR when $\rho_x=1$.
The gap to the performance upper bound as a function of $\rho_x$ is
plotted for DQLC and uncoded transmission in Fig.~\ref{fig:loss_opta_rho}, for 40dB channel SNR. Note that the distance to the upper bound is largest for DQLC when $\rho_x$ is around $0.9$, and that DQLC and uncoded transmission both reach the upper bound in the limit
$\rho_x\rightarrow 1$. Note that DQLC performs better than uncoded transmission for most $\rho_x$ values at 40dB SNR.

If there is a demand for equal transmit power from each encoder, one can still use DQLC with timesharing. I.e. each encoder described here is used on each source 1/3 of the time (1/M in general). As shown for the $M=2$ case in~\cite{Floor_Kim_Wernersson11_TCOM}, a further backoff from the bound compared to that in Fig.~\ref{fig:Results_M3_Rc} must then be expected.

\section{Summary and extensions}\label{sec:summary}
In this paper, a distributed delay-free and low complexity joint
source-to-channel mapping was proposed and bounds were derived for transmission of a multivariate Gaussian over a Gaussian MAC. Both linear and nonlinear
mappings were analyzed. A linear mapping (uncoded transmission) achieves the
performance upper bound within a certain range of \emph{low} channel SNR. A
nonlinear mapping, DQLC, was introduced to improve on uncoded transmission at higher SNR. A collaborative scheme (Shannon-Kotel'nikov mapping) was also introduced to provide an approximate bound for zero delay schemes.

DQLC does not achieve the performance upper
bound, but leaves a certain gap which value depends on both correlation and the number of sources.
However, DQLC constitutes a constant gap to the bound in any case as SNR$\rightarrow \infty$ and its received fidelity (SDR) improves with increasing correlation without changing
the basic encoder and decoder structure. DQLC therefore outperforms uncoded transmission as the SNR gets high enough in any case. Unfortunately, the loss to the bound increases somewhat with the number of sources (something the collaborative S-K mappings does not), but is fortunately bounded to a
finite value as the number of sources goes to infinity. Optimization and simulation for 3 sources also showed that the gap to the bound decreases somewhat as the SNR drops.

It is also important to note that DQLC can be applied for any unimodal source (and channel) distribution and optimized using the same method as presented in this paper.

Future research should aim at finding, if possible, a scheme where the loss to the bound does not increase with the number of sources, as is the case for the collaborative S-K mappings. The generalization of DQLC to arbitrary code length should also be investigated. Recently it has been shown that such a generalization achieves the performance upper bound when SNR is high for any number of uncorrelated sources~\cite{Floor_kim_ramstad_ITW2012}. What remains is to prove what happens for arbitrary correlation and SNR. The impact of practical issues like imperfect timing and synchronization should also be addressed in order to get a step closer to a possible practical realization.

\section*{Acknowledgment}
This work was funded by the Norwegian Research Council under
project MELODY (187857/S10), the Swedish Research Council and
VINNOVA.

\appendix
Let $\mathbf{m}_y$ denote the mean and $\Sigma$ the
covariance matrix of an n-dimensional Gaussian random vector $\mathbf{y}$
\begin{theorem}\label{th:multvar_cond_moments}~\cite[p. 12]{muirhead82}
Let $\mathbf{y}\sim\mathcal{N}_n(\mathbf{m}_y,\Sigma)$ and make a
partition of $\mathbf{y}$ into a $k\times 1$ vector $\mathbf{y}_a$
and a $(n-k)\times 1$ vector $\mathbf{y}_b$. Then make the following
partition of  $\mathbf{m}_y$ and $\Sigma$:
\begin{equation}\label{e:subdiv_matrix}
\mathbf{m}_y=\begin{bmatrix} \mathbf{m}_a\\ \mathbf{m}_b
\end{bmatrix},\Sigma=\begin{bmatrix}
\Sigma_{aa}&\Sigma_{ab}\\ \Sigma_{ba}&\Sigma_{bb}
\end{bmatrix}
\end{equation}
Further, let $\Sigma_{bb}^{+}$ 
be a matrix satisfying
 $\Sigma_{bb}\Sigma_{bb}^{+} \Sigma_{bb}=\Sigma_{bb}$
and let $\Sigma_{aa \cdot b}=\Sigma_{aa}-\Sigma_{ab}\Sigma_{bb}^+ \Sigma_{ba}$.
Then,
\begin{equation}\label{e:gen_cond_pdf}
p(\mathbf{y}_a|\mathbf{y}_b) \sim \mathcal{N}_k (\mathbf{m}_a+\Sigma_{ab}\Sigma_{bb}^+ (\mathbf{y}_b-\mathbf{m}_b),\Sigma_{aa\cdot b})
\end{equation}
\end{theorem}

\begin{proof}
See~\cite[pp. 12-13]{muirhead82}
\end{proof}


\bibliographystyle{IEEEtran}
\bibliography{references}

\end{document}